\documentclass[aps,letter,prx,reprint,superscriptaddress,notitlepage,floatfix,nolongbibliography]{revtex4-2}

\usepackage[english]{babel}

\usepackage[letterpaper,top=2cm,bottom=2cm,left=1.75cm,right=1.75cm,marginparwidth=1.75cm]{geometry}
\usepackage{ragged2e}

\usepackage{amsmath}
\usepackage{graphicx}
\usepackage[dvipsnames]{xcolor}
\usepackage[colorlinks=true, allcolors=OliveGreen]{hyperref}
\usepackage{wrapfig}

\usepackage{booktabs}
\usepackage{colortbl}
\usepackage{array}

\usepackage{amsfonts}
\usepackage{bm}
\usepackage{amsthm}
\usepackage{chemformula}

\usepackage[most]{tcolorbox}
\newtcolorbox[auto counter]{ProofBox}[2][]{%
  breakable,
  colback=blue!4,
  colframe=blue!40!black,
  coltitle=white,
  title={Box~\thetcbcounter:\ #2},
  fonttitle=\bfseries\sffamily,
  title after break={},
  #1
}

\usepackage{etoolbox}   

\newcounter{savedfig}
\AtBeginEnvironment{ProofBox}{%
  \setcounter{savedfig}{\value{figure}}%
  \setcounter{figure}{0}%
}
\AtEndEnvironment{ProofBox}{%
  \setcounter{figure}{\value{savedfig}}%
}

\newtcbtheorem[no counter]{Proof}{Box}{
  enhanced,
  sharp corners,
  attach boxed title to top left={
    yshifttext=-1mm
  },
  colback=white,
  colframe=blue!25,
  fonttitle=\bfseries,
  coltitle=black,
  boxed title style={
    sharp corners,
    size=small,
    colback=blue!25,
    colframe=blue!25,
  },
}{prf}

\newtheorem{lemma}{Lemma}

\definecolor{hlgreen}{HTML}{C8E6C9}
\definecolor{darkgreen}{HTML}{4CAF50}
\newcommand{\hl}[1]{\colorbox{hlgreen}{$#1$}}

\begin{document}
\title{Informational blueprints reveal condition-dependent gene regulatory architectures}

\author{Doruk Efe Gökmen}
\affiliation{NSF-Simons National Institute for Theory and Mathematics in Biology, Chicago, IL 60611, USA}
\affiliation{James Franck Institute, The University of Chicago, Chicago, IL 60637, USA}
\author{Rosalind Wenshan Pan}
\affiliation{Division of Biology and Biological Engineering, California Institute of Technology, Pasadena, CA 91125, USA}
\author{Tom Röschinger}
\affiliation{Division of Biology and Biological Engineering, California Institute of Technology, Pasadena, CA 91125, USA}
\author{Stephen Quake}
\affiliation{Department of Bioengineering, Stanford University, Stanford, CA 94305, USA}
\affiliation{Department of Applied Physics, Stanford University, Stanford, CA 94305, USA}
\author{Hernan G Garcia}
\affiliation{Departments of Molecular \& Cell Biology and of Physics, Biophysics Graduate Group,  Institute for Quantitative Biosciences-QB3, University of California, Berkeley, CA, USA; Chan Zuckerberg Biohub--San Francisco, San Francisco, California, USA}
\author{Rob Phillips}
\affiliation{Division of Biology and Biological Engineering, California Institute of Technology, Pasadena, CA 91125, USA}
\affiliation{Department of Physics, California Institute of Technology, Pasadena, CA 91125, USA}
\author{Vincenzo Vitelli}
\email{vitelli@uchicago.edu}
\affiliation{James Franck Institute, The University of Chicago, Chicago, IL 60637, USA}
\affiliation{Leinweber Institute for Theoretical Physics, The University of Chicago, Chicago, IL 60637, USA}
\affiliation{Chan Zuckerberg Biohub--Chicago, Chicago, Illinois, USA}

\begin{abstract}
    {
    While coding regions in the genome have a direct interpretation in terms of protein products, significant fractions are non-coding and yet control essential biological functions. Unlike the genetic code, there is no ``lookup table'' that identifies where regulatory proteins, known as transcription factors (TFs), bind. Here, we extract these binding sites by distilling sequences of nucleotide letters into collective coordinates (hyperletters) representing the binding sites that are active under specific environmental conditions. Going beyond local information footprints between individual bases and expression levels, our \emph{information blueprint} algorithm compresses the global information by optimising filters that simultaneously scan an entire promoter sequence. Inspired by renormalisation-group techniques, we identify TF binding sites as coarse-grained variables combining groups of correlated mutations with the highest collective impact on gene expression. We validate our approach on experimental data for {\it E. coli} and  discover novel regulatory elements illustrating its deployment at scale across growth conditions.  
    }

\end{abstract}
\maketitle

The pioneers of molecular biology mused over the two great polymer languages~\cite{crick_genetic_1966} which life uses to encode information and biological function: nucleic acids and proteins.
While the central dogma~\cite{crick_central_1970} yielded  a dictionary relating the coding regions to the transcribed protein products,
a significant fraction of genomes is non-coding: in part, it serves the regulatory function of adapting the activity of genes to varying cellular environments.

Indeed, one of the grand challenges of the genomic era is to discover the genome-wide architecture of transcriptional regulation~\cite{Keseler2011, Keseler2021EcoCyc,Koo2020}.  That is, to find the constellation of binding sites which control the multitude of genes within a given organism, as well as the transcription factors that bind those sites.  

\begin{figure}[hbt!]
    \centering
    \includegraphics[width=\linewidth]{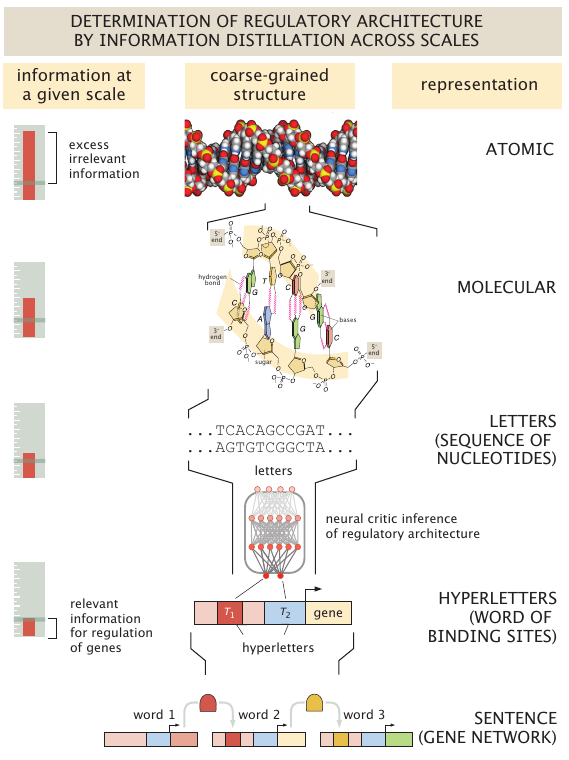}
    \caption{\textbf{From a nucleotide sequence to a constellation of binding sites.}  
    The spatial structure of the DNA double helix can be represented  as a one-dimensional sequence of nucleotide letters (A, T, G, C). This is an act of coarse graining, one that discards the polymer-level physical degrees of freedom while retaining the genetic  information encoded in the base pair ordering. Here we take this  one step further: we construct \emph{hyperletters} from groups of nucleotide letters, each encoding the presence of a transcription factor binding site, and coarse grain the promoter sequence into a short vector $\boldsymbol{T}$ of such hyperletters. The result is a low-dimensional, interpretable representation of the regulatory architecture of a promoter.
    }
    \label{fig:fig1}
\end{figure}

The appropriate protein vocabulary to describe such transcriptional regulation is that of transcription factors (TFs), which are proteins that bind to specific {\it cis}-regulatory elements—mesoscopic sequence motifs in the DNA that control genes.
While the mapping from non-coding DNA to activities of transcribed genes is a complicated process that is not yet understood in quantitative and predictive detail, 
the combinatorial occupancy of TF binding sites is a natural collective variable that determines the transcriptional activity~\cite{spitz_transcription_2012}.
This suggests a coarse-graining approach to model the non-coding genome at the level of binding configurations of regulatory proteins, 
rather than at the nucleotide level—a natural separation of scales between the ~1 bp resolution of sequencing and the typical ~20 bp size of a regulatory protein binding site along a bacterial genome.

A principal tool for attacking this problem is the ``mutate and readout" strategy of massively parallel reporter assays (MPRAs)~\cite{Schneider1989,patwardhan_high-resolution_2009,Kinney2010Final,Sharon2012,Kosuri2013,Urtecho2019, Lagator2022,Belliveau2018}. By systematically mutating a promoter and measuring the resulting gene expression, these assays generate a high-dimensional mapping from nucleotide microstates to expression levels. For example, 
such datasets have been collected across hundreds of different bacterial promoters over tens of different environmental conditions~\cite{Roeschinger2026}. However, while we have a clear scheme for generating this data, we lack a systematic and rigorous conceptual framework for interpreting it.

Efforts to decipher regulatory architecture from sequence have followed several complementary paths. Classical bioinformatic approaches use motif discovery, statistical overrepresentation, comparative genomics, and phylogenetic footprinting to identify candidate regulatory elements~\cite{Lawrence1993,Bailey1994,Stormo2000,bussemaker_building_2000,van_nimwegen_probabilistic_2002,Sinha2002,Wasserman2004,Tompa2005}. These methods have been enormously useful, but they usually begin from collections of co-regulated or homologous sequences and return candidate motifs rather than a direct condition-dependent mapping from promoter sequence to expression. At the other extreme, modern machine-learning models fit high-capacity sequence-to-function maps that can predict chromatin state, transcription-factor binding, or gene expression with impressive accuracy~\cite{Zhou2015DeepSEA,Kelley2018Basenji,Avsec2021BPNet,Avsec2021Enformer,avsec_advancing_2026,hu_multiscale_2025,barbadillamartinez_predicting_2025,Mitra2024,deAlmeida2022,seitz_interpreting_2024}. Although these models are powerful, relating their internal representations to compact mechanistic descriptions of regulation remains an open challenge~\cite{Koo2020,seitz_interpreting_2024}. Here, we seek a middle ground~\cite{Tareen2020,Koo2020,Lally2025}: we approach the complexity of the data gradually by progressively incorporating relevant information, a procedure known in physics contexts as coarse graining. By systematically eliminating irrelevant variables, we distill the sequence into interpretable filters and logic circuits that are flexible enough to capture cooperative and long-range interactions.

To do this, we transform the concept of the information footprint~\cite{Kinney2010Final,Belliveau2018,Ireland2020}---which locally identifies informative bases in isolation---into a global \emph{information blueprint}: a systematic coarse graining principle that identifies how correlated binding sites cooperatively determine regulatory logic. Starting from nucleotide sequences with point mutations, we use an optimal lossy compression scheme~\cite{infbottle1, PhysRevLett.127.240603} to identify lower-dimensional collective coordinates that are relevant for gene expression, as shown schematically in Fig.~\ref{fig:fig1}. The notion of compression here is analogous to moving from the coordinates of $10^{26}$ atoms in a solid to a few collective vibrational modes. Extracting these relevant collective variables accomplishes the essential task of annotating functional binding sites—distinguishing regulatory signal from noise. Moreover, the method also yields insights that go deeper, revealing how a promoter integrates signals from multiple transcription factors, extracting cooperative interactions and non-local effects such as DNA looping.

Our method seeks to find meaning through correlations between sequence and level of expression. As a result, our approach is inherently condition dependent, revealing which binding sites are functionally active under specific environmental conditions. 
Indeed, the same promoter sequence can exhibit distinct regulatory architectures depending on the growth environment, \emph{e.g.}, a repressor site that is invisible under oxidative stress may dominate expression control under glucose. This condition-dependence emerges naturally from our framework, where environmental signals serve as essential inputs to the regulatory logic circuit.

The remainder of this paper is organized as follows. In Section~\ref{sec:the_method}, we define our \emph{information blueprint} principle for identifying binding sites via the optimal lossy compression of genomic information. In Section~\ref{sec:synthetic_validation}, we validate this approach using synthetic datasets, demonstrating that our method correctly recovers known architectures—from simple repression to DNA looping—without prior assumptions. In Section~\ref{sec:experimental}, we apply the framework to experimental MPRA data for \emph{E. coli}, recovering established binding sites in well-characterized promoters across growth conditions and making testable predictions for unannotated genes. Finally, we conclude with an outlook on how this coarse-graining procedure could be iterated to infer genome-wide regulatory networks.

\section{Coarse graining sequences to binding configurations}
\label{sec:the_method}

A typical bacterial promoter like those studied here spans roughly 150 bp, yet only a small fraction of positions—those within TF binding sites—significantly impact expression. Massively parallel reporter assays (MPRAs)~\cite{patwardhan_high-resolution_2009,Kinney2010Final, melnikov_systematic_2012, patwardhan_massively_2012} systematically probe this by measuring expression across libraries of point mutants.
Expression varies widely across an MPRA library of many mutant sequences—some mutations reduce expression, others increase expression and yet others have little effect. 
The challenge is to (i) discover biologically interpretable collective variables encoding binding sites in a gene regulatory architecture and (ii) establish how these sites are correlated. In the next section we present the general theory behind our approach and refer to the Box where a pedagogical example (that may be read in parallel) is worked out step by step.

\onecolumngrid
    
\begin{ProofBox}[label={box:algorithm}]{The information-blueprint method in action}
\justifying
We walk through the simplest application of our method to
a toy promoter as shown in Figure~\ref{fig:FiltersBox} with $N = 5$ positions regulated only by RNA polymerase with a binding site at position $x = 4$, and a toy MPRA library with a $20\%$ mutation rate. While here we provide a basic outline of the implementation of the information blueprint approach in the context of this toy model, all calculations are presented in detail in Appendix~\ref{Appendix:BoxCalculations}. Our dataset consists of five sequences each carrying exactly one mutation. As a result the $i$-th position of the vector $\mathbf{B}^{(m)}$ describing the $m$-th sequence is given by
\begin{equation}
    B_i^{(m)} = \delta_{im}.
    \tag{b1}
\end{equation}
Expression is high ($\mu^{(m)}=$h) for all mutant sequences except $m = 4$, whose mutation coincides with the RNAP binding site ($\mu^{(4)}=$l).
We adopt a simple filter parametrization localised at position $x$: $\Lambda_i(x) = \delta_{ix}$. The hyperletter then associated with the position $x$ for the $m$-th sequence is given by 
\begin{equation}\label{eq:box_1}
    T^{(m)}(x) = \sigma\left(\sum_{i=1}^5 B_i^{(m)} \Lambda_i(x) \right) = \sigma\left(\sum_{i=1}^5 \delta_{im} \delta_{ix}\right) = \sigma(\delta_{mx})=\begin{cases} 1, \quad m=x, \\ 0, \quad m\neq x \end{cases},\tag{b2}
\end{equation} 

\begin{wrapfigure}{r}{0.55\linewidth}
\vspace{-16pt}
\centering
\includegraphics[width=\linewidth]{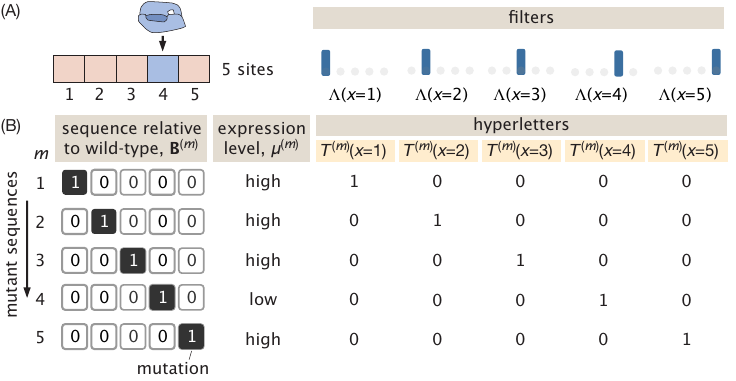}
\vspace{-12pt}
\caption{The information blueprint method for a toy promoter. (A) A five basepair sequence in which the RNAP binds at site four.  The filters are defined by the position at which they are non-zero. (B) Five mutated sequences, their level of expression and the hyperletters that go with each filter for each sequence. }
\label{fig:FiltersBox}
\end{wrapfigure}
\noindent where $\sigma(x)$ is a nonlinear function that imposes a thresholding.
The table shows the value of each hyperletter for each sequence $T^{(m)}$ in the MPRA library.  Each filter flags a different mutant in the set of mutated sequences.

\noindent

We seek to determine which of the possible five hyperletters is most informative for predicting gene expression. That is, we want to find the value of $x$ that maximizes the mutual information between sequence and hyperletter as
\begin{equation}
    I\left(\mu : T(x)\right) = H(\mu) - H\left(\mu | T(x)\right).
    \tag{b3}
\end{equation}
The entropy term $H(\mu)$ can be calculated using the definition from equation~\ref{eq:entropy}. More importantly, this term does not depend on the choice of hyperletter $T(x)$. What depends on $x$ is the conditional entropy $H(\mu | T(x))$. For example, for the $x=4$ hyperletter this conditional entropy is given by
\begin{equation}\label{eq:box_2}
    H(\mu | T(x=4)) \,=\, \underbrace{P(T(x=4)=1)}_{1/5} \cdot \underbrace{H(\mu | T(x=4)=1)}_{0} \,+\, \underbrace{P(T(x=4)=0)}_{4/5} \cdot \underbrace{H(\mu | T(x=4)=\texttt{0})}_{0} \,=\, 0. \tag{b4}
\end{equation}
To see why the entropy terms in the equation vanish, we can, for example, focus on
\begin{equation}
\begin{aligned}
    H(\mu | T(x=4)=1) = - & \left[ P\left(\mu = h | T(x=4)=1\right) \, \log_2 P\left(\mu = h | T(x=4)=1\right) \right. \\
    & \left. + P\left(\mu = l | T(x=4)=1\right) \, \log_2 P\left(\mu = l | T(x=4)=1\right) \right] \\
    & - \left[ 0 \log_2(0) + 1 \log_2(1) \right] = 0.
    \end{aligned}
    \tag{b5}
\end{equation}
As a result, the mutual information between gene expression and the $T(x=4)$ hyperletter is given by
\begin{equation}
    I\left(\mu : T(x=4)\right) = H(\mu) - H\left(\mu | T(x=4)\right) = H(\mu) \approx  0.722.
    \tag{b6}
\end{equation}

\noindent For any $x \neq 4$, by contrast the conditional entropy is given by
\begin{equation}\label{eq:box_3}
    H\left(\mu | T(x \neq 4)\right)\, =\, \underbrace{P\left(T(x \neq 4)=\texttt{1}\right)}_{1/5} \cdot \underbrace{H\left(\mu | T(x \neq 4)=\texttt{1}\right)}_{0}\, + \,\underbrace{P\left(T(x \neq 4\right)=\texttt{0})}_{4/5} \cdot \underbrace{H\left(\mu | T(x \neq 4)=\texttt{0}\right)}_{\tfrac{1}{4}\log 4 + \tfrac{3}{4}\log\tfrac{4}{3}} \,>\, 0. \tag{b7}
\end{equation}
As a result, the mutual information for any $x \neq 4$ is
\begin{equation}
    I\left(\mu : T(x \neq4)\right) = H(\mu) - H\left(\mu | T(x \neq 4)\right) = H(\mu) \approx  0.073.
    \tag{b8}
\end{equation}
The mutual information 
therefore peaks sharply at $x = 4$, saturating the upper bound $H(\mu) \approx 0.722$\,bits, while every other filter position yields only a small residual signal $\approx 0.073$\,bits, as shown in  the left panel of Figure~\ref{fig:box_mi}. 
The binding site emerges as the position where a single bit is most informative of expression.

\begin{wrapfigure}{l}{0.4\linewidth}
    \vspace{-10pt}
    \includegraphics[width=\linewidth]{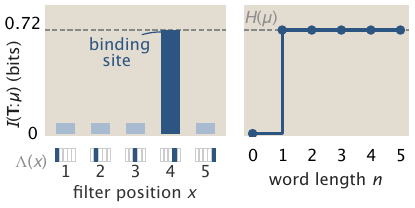}
    \vspace{-15pt}
    \caption{Finding the most informative coarse-graining of the sequence. (A) Scores for the five competing filters.  (B) Mutual information score as a function of word length (i.e. number of hyperletters).}
    \label{fig:box_mi}
\end{wrapfigure}
This single-filter solution also tells us when to stop adding filters. Since $I(T:\mu) \leq H(\mu)$ for any compression $T$, and one filter at the binding site already saturates this information-theoretic bound with $H(\mu |T) = 0$, every additional filter contributes exactly zero information. The curve $I(\mu:T)$ versus word length $n$ therefore plateaus immediately after $n = 1$, with $\Delta I(n \to n{+}1) = 0$ for all $n \geq 1$. This example thus suggests a natural stopping criterion for the word length: we halt when the marginal gain $\Delta I(n \to n{+}1)$ falls below the statistical noise floor, correctly selecting $n = 1$ here.

\noindent As shown in Fig.~\ref{fig:algorithm}, the same logic applies to the full $N=150$~bp constitutive promoter, demonstrating that maximizing $I(T:\mu)$ over all filters thus discovers the binding site without any prior knowledge of its location.
\end{ProofBox}

\twocolumngrid

\subsection{A regulatory alphabet: mapping sequences into hyperletters and words}
\label{sec:words}

\begin{figure*}[hbt!]
    \includegraphics[width=0.85\linewidth]{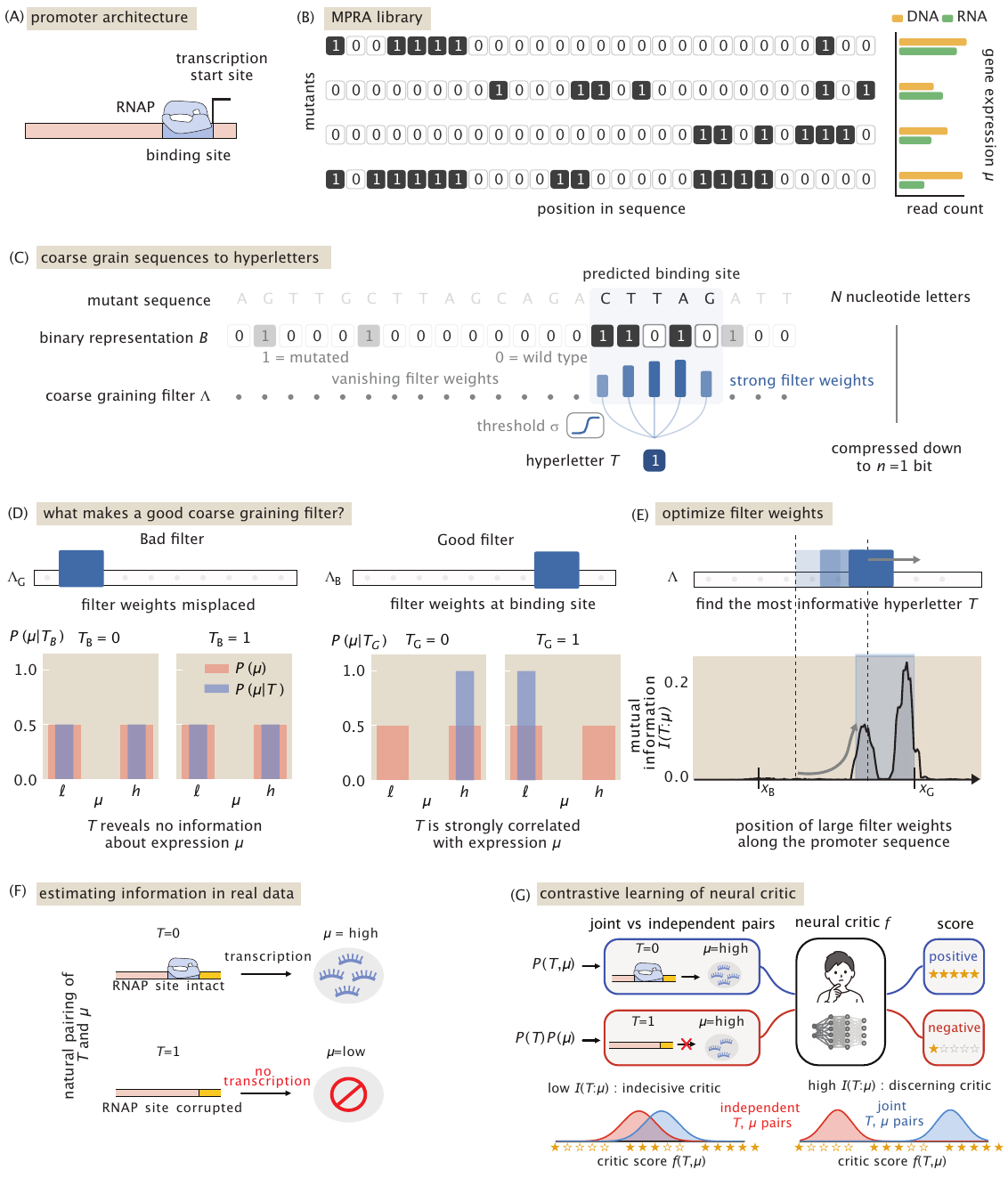}
    \caption{\textbf{Finding the binding site by optimal compression.}
    (A)~Schematic of a constitutive promoter with a known RNAP binding site.
    (B)~A toy MPRA library: each row is a mutant sequence; bar plot gives the DNA and RNA read counts whose ratio defines expression~$\mu$.
    (C)~A linear filter $\bm{\Lambda}$ is used to compress each sequence into a word $T\!\in\!\{\texttt{0},\texttt{1}\}$. Two filters are compared: $\bm{\Lambda}_{\rm G}$ sits on the RNAP site, $\bm{\Lambda}_{\rm B}$ misses it.
    (D)~The table lists expression and both compressed words for every mutant. Sliding a filter across the promoter and computing $I(T:\mu)$ at each position yields a trace that peaks at the RNAP site. Conditional distributions $P(\mu| T)$ obtained by counting from the table. The good filter $T_{\rm G}$ sharply distinguishes high from low expression (blue vs.\ red), yielding large $I(T_{\rm G}:\mu)$; the bad filter $T_{\rm B}$ leaves the conditionals nearly unchanged. Maximizing $I(T:\mu)$ over filters therefore discovers the binding site without prior knowledge.
    (E)~In practice, estimating mutual information by binning and counting becomes unreliable for continuous expression data and large sequence spaces. The solution is to turn estimation into a classification problem based on the underlying biological correlations between word $T$ and expression $\mu$: an intact site ($T\!=\!0$) permits transcription (high $\mu$), a corrupted site ($T\!=\!1$) abolishes it, producing correlated $(T,\mu)$ pairs from the joint $P(T,\mu)$.
    (F)~A neural critic $f$ is trained to distinguish joint pairs from independently shuffled ones drawn from $P(T)P(\mu)$. The gap between the critic's scores on the two sample types provides a lower bound on $I(T:\mu)$: an opinionated critic (right) results in high mutual information \emph{cf.} top histograms in (D), an indecisive one (left) yields a low value, \emph{cf.} bottom pair of histograms in (D).
}\label{fig:algorithm}
\end{figure*}

An MPRA library is built up of $M$ mutant promoters of length $N$, where $M$ is typically on the order of thousands and $N$ on the order of hundreds.
Each of these promoters, in turn, has typically 10\% of its bases mutated.
We create a simplified version of this situation in Box~\ref{box:algorithm} by considering a toy promoter of length $N=5$ where only one site is mutated (20\% mutation rate), hence generating a library that consists of $M=5$ unique mutants.
In general, for each mutant sequence labelled by $m$, \emph{e.g.}, \texttt{ATGT\underline{A}...} where underlined bases differ from the wild type \texttt{ATGTT...}, we construct a binary vector $\bm{B}^{(m)}$ with entries $B_i^{(m)} = 1$ at mutated positions and $0$ otherwise--so that $\bm{B}^{(m)} = [0,0,0,0,1,\cdots]$ as shown in Figure~\ref{fig:algorithm}~(B). 
In Box~\ref{box:algorithm}, $\bf{B}$ is simply a 5-dimensional vector.
The corresponding expression level $\mu^{(m)}$ for each sequence is given by the ratio between RNA and DNA counts for that given promoter variant.
For the toy promoter example in Box~\ref{box:algorithm} (see Fig.~\ref{fig:FiltersBox}), the expression vanishes only when the mutation coincides with the single-base RNAP site, and is 1 otherwise.
The choice of normalizing expression by DNA counts helps to correct for biases in library construction and sequencing, ensuring that the measured expression levels reflect true biological differences rather than technical artifacts.

Just as the nucleotide alphabet $\{\texttt{A},\texttt{C},\texttt{G},\texttt{T}\}$ represents a microscopic coarse-graining of the underlying atomic complexity of nucleotides, our method takes this abstraction one step further.
Each mutant sequence $\bm{B}^{(m)}$ in the library will be assigned a short compressed word $\bm{T}^{(m)}$, which is what we mean by collective coordinates.
Each word will be a vector, whose components denoted by $T_\nu^{(m)}$ can adopt one of two values: 
either \texttt{0} or \texttt{1}. We call these components hyperletters to explicitly distinguish them from standard genomic letters. 
We will use the roman indices $i$ to label individual base pairs, and greek indices $\nu$ to label the collective coordinates.
While at this stage we have not yet established the meaning of the collective coordinates, our goal is to have $\nu$ label the binding sites.
While nucleotide states $B_i^{(m)}$ encode chemical identity, hyperletters $T_\nu^{(m)}$ will encode the coarse-grained functional state of the regulatory architecture (\emph{e.g.}, intact vs disrupted binding sites).

How do we distill such hyperletters from sequences? 
We construct a biologically inspired algorithm consisting of three steps.
First, we make an initial guess (to be optimized iteratively) for the number of putative regulatory elements, which we denote by $n$.  For example, for the constitutive promoter in Box~\ref{box:algorithm}, we make a guess that there will only be one hyperletter corresponding to the RNAP binding site.
The next step mimics how a regulatory protein scans the DNA to recognize a specific binding site. In our algorithm, an $n\times N$ linear filter $\Lambda_{\nu i}$ (playing the role of the protein) is a matrix acting on the sequence vectors $\bm{B}^{(m)}$ in the MPRA library. 
In Box~\ref{box:algorithm}, since we use 1 putative hyperletter, the filter is simply a $1\times5$ matrix.
Algorithmically, the scanning action of $\Lambda_{\nu i}$ consists of tracking when mutations accumulate within a binding site until the sequence becomes unrecognizable and the transcription factor can no longer bind.  
To mimic how the disruption of the protein binding depends on a threshold: we feed the output of the previous step to a nonlinear function $\sigma$ (a sigmoid), which effectively assigns \texttt{0} if the site is still recognisable or \texttt{1} to $T_{\nu}^{(m)}$, when the site is disrupted.
The actual value of this threshold is not known a priori but it is instead learned, and differs across binding sites and promoters depending on how strongly each mutation changes the binding energy.

These three steps can be combined in the following graphical equation, defining the hyperletter for each putative regulatory element $\nu = 1, \ldots, n$:
\begin{equation}
{\scriptstyle
\begin{bmatrix}
    \hl{\Lambda_{1,1}} & \hl{\cdots} & \hl{\Lambda_{1,N}} \\
    \vdots & \ddots & \vdots \\
    \Lambda_{n,1} & \cdots & \Lambda_{n,N}
\end{bmatrix}}
\cdot
{\scriptstyle 
\begin{bmatrix}
    \hl{B_1^{(m)}} \\ \hl{\vdots} \\ \hl{B_N^{(m)}}
\end{bmatrix}}
\;\xrightarrow{\;\sigma(\cdot)\;}
{\scriptstyle
\begin{bmatrix}
    \hl{T_1^{(m)}} \\ \vdots \\ T_n^{(m)}
\end{bmatrix}}
\label{eq:compression}
\end{equation}
or more compactly, 
\begin{align}\label{eq:linear}
    T_\nu^{(m)} = \sigma\left(\sum_{i=1}^N \Lambda_{\nu i} \, B_i^{(m)} \right).
\end{align}
Taken together, the $n$ filters map each full-length sequence to a short word of hyperletters, \emph{e.g.} for $n = 3$, 
\begin{equation*}
    \underbrace{\texttt{A\,C\,G\,T\,T\,\ldots}}_{i = 1, \ldots, N} \;\mapsto\; \underbrace{\texttt{0\,1\,1}}_{\nu = 1, 2, 3} .
\end{equation*}
We will call the collection of filters $\Lambda_{\nu i}$ the \emph{informational blueprint} of the promoter.
The next sections show how filters, constructed by optimally compressing the information the sequences carry about expression, naturally localize on distinct binding sites. Such filters will have significant weights only at positions where mutations disrupt binding of RNAP or TFs.

\subsection{Shannon entropy for MPRA libraries} 
In view of the complexity of typical MPRA libraries, it is useful to 
represent these large data sets probabilistically.
An MPRA library of $M$ mutant sequences can be viewed as a set of independent samples from a joint distribution $P_{\bm{B}, \mu}\left(\bm{B}^{(m)}, \mu^{(m)}\right)$ over sequence and expression.
Drawing a sequence at random from the library, the probability of observing a given expression level is estimated by the fraction of sequences that fall in each expression bin: $P_\mu\left(\mu^{(m)}\right) \approx \text{count}\left(\mu^{(m)}\right) / M$.
The variability of expression across the library is then quantified by the Shannon entropy,
\begin{equation}\label{eq:entropy}
    H(\mu) = - \sum_{m=1}^{M} \, P_\mu\left(\mu^{(m)}\right)\log_2 P_\mu\left(\mu^{(m)}\right),
\end{equation}
which measures, in bits, the baseline uncertainty about the expression level of a randomly drawn sequence---before any information about its sequence is taken into account.
For the simple example in Box~\ref{box:algorithm}, the entropy can be readily calculated by counting the bins 0 (gene off) and 1 (gene on) in the table of Fig.~\ref{fig:FiltersBox}.
More efficient estimation strategies of information theoretic quantities by machine learning avoid binning altogether, see Methods and Fig.~\ref{fig:algorithm}~(F).

Applying Eq.~\ref{eq:linear} to every sequence in the library maps each sequence to a compressed word
and yields a joint distribution $P_{\bm{T}, \mu}\left(\bm{T}^{(m)}, \mu^{(m)}\right)$.
If the filters are well chosen, knowing $\bm{T}^{(m)}$ should sharply reduce the uncertainty about expression---that is, $H(\mu |\bm{T}) \ll H(\mu)$.  Which filters achieve such an optimal compression?

\subsection{Optimally compressed words}
\label{sec:optimalwords}

To assess how well we compress information, let's play the following game.
A sequence $\bm{B}^{(m)}$ is drawn at random from the MPRA library, and you must guess its expression level $\mu^{(m)}$.
If no information about the sequence is revealed, 
the success of your guessing strategy is limited by the uncertainty $H(\mu)$ in Eq.~\ref{eq:entropy}.
Now suppose that, before making this prediction, you are permitted to ask a small number of yes-or-no questions about the sequence (recall it is a bit string). Of course,  you are not allowed to ask the full $N=150$ binary queries it would take to reconstruct every letter in the full sequence, only a small number $n$, where $n\ll N$.
Each yes or no answer, constitutes one hyperletter $T_\nu^{(m)}=\texttt{1}$ or $\texttt{0}$ of the compressed word $\bm{T}^{(m)}$.

The compression challenge is thus to find which queries are most useful.
Most positions have no bearing on expression, thus the queries would be wasted if they are not about a few carefully selected positions.
If your queries are well chosen, a handful of answers suffices. For example, learning that the repressor binding site is intact and the activator site is disrupted already constrains the expected expression level considerably.
The compressed word $\bm{T}$ is the list of answers to these optimally selected queries, and the residual uncertainty after learning $\bm{T}^{(m)}$ is the conditional entropy $H(\mu|\bm{T})$.
The reduction in uncertainty---how much the prediction sharpens upon learning $\bm{T}^{(m)}$---is the mutual information,
\begin{equation}\label{eq:information_blueprint}
    I(\bm{T} : \mu) = \underbrace{H(\mu)}_{\text{initial uncertainty}} - \underbrace{H(\mu |\bm{T})}_{\text{remaining uncertainty}}.
\end{equation}
The concrete task is therefore to find the filter $\Lambda_{\nu i}$ (defined in Eq.~\ref{eq:linear}) that maximizes the mutual information $I$ and hence
satisfies the optimality condition
\begin{equation}\label{eq:objective}
    \delta  I(\bm{T} : \mu) = 0,
\end{equation}
subject to the constraint $n \ll N$.
For the toy example in Box~\ref{box:algorithm}, Eq.~\ref{eq:information_blueprint} can be obtained in a closed form as a function of which query is made, whereby we find that the optimal query is indeed the one that checks if the RNAP site is mutated.

The few hyperletters must therefore account for as much of the variability in expression as possible, as illustrated in Fig.~\ref{fig:algorithm}~(D) by the histograms of the distribution of expression conditioned on a given state of the word, $P(\mu | \bm{T})$. 
Thanks to the biologically informed compression, the hyperletters $T_\nu^{(m)}$ acquire concrete regulatory meaning. In the Box, we carry out this procedure in the simplest of cases where it can actually be done by hand.

In practice, estimating $I(\bm{T}:\mu)$ from finite data by binning and counting, as was done in the Box, becomes unreliable as the dimensionality of $\bm{T}$ grows.
We instead turn to improved estimation techniques using neural networks~\cite{mine, infonce, poole2019variationalboundsmutualinformation} and rooted in large-deviations theory in statistics~\cite{donsker_varadhan_1983}.
The idea is to recast estimation as a classification task: a neural network critic $f$ receives a pair $(\bm{T}, \mu)$ and must judge whether the two come from the natural joint distribution $P(\bm{T},\mu)$---where an intact site ($T\!=\!0$) yields high expression and a corrupted site ($T\!=\!1$) low expression, as shown in Fig.~\ref{fig:algorithm}(E)---or from an independent shuffle $P(\bm{T})P(\mu)$ in which such correlations are destroyed.
When $\bm{T}$ carries regulatory information a discerning critic can reliably tell the two apart, yielding a large score gap; when $\bm{T}$ is uninformative the pairs look identical and the gap vanishes, as illustrated in Fig.~\ref{fig:algorithm}(F).
This score gap is a differentiable lower bound on $I(\bm{T}:\mu)$, which lets us optimise the compression map $\Lambda$ directly by gradient descent~\cite{PhysRevLett.127.240603, PhysRevE.104.064106, gokmen2024compression}, using the mutual information itself as the objective (see Methods for details).

\begin{figure*}[hbt!]
    \centering
    \includegraphics[width=0.85\linewidth]{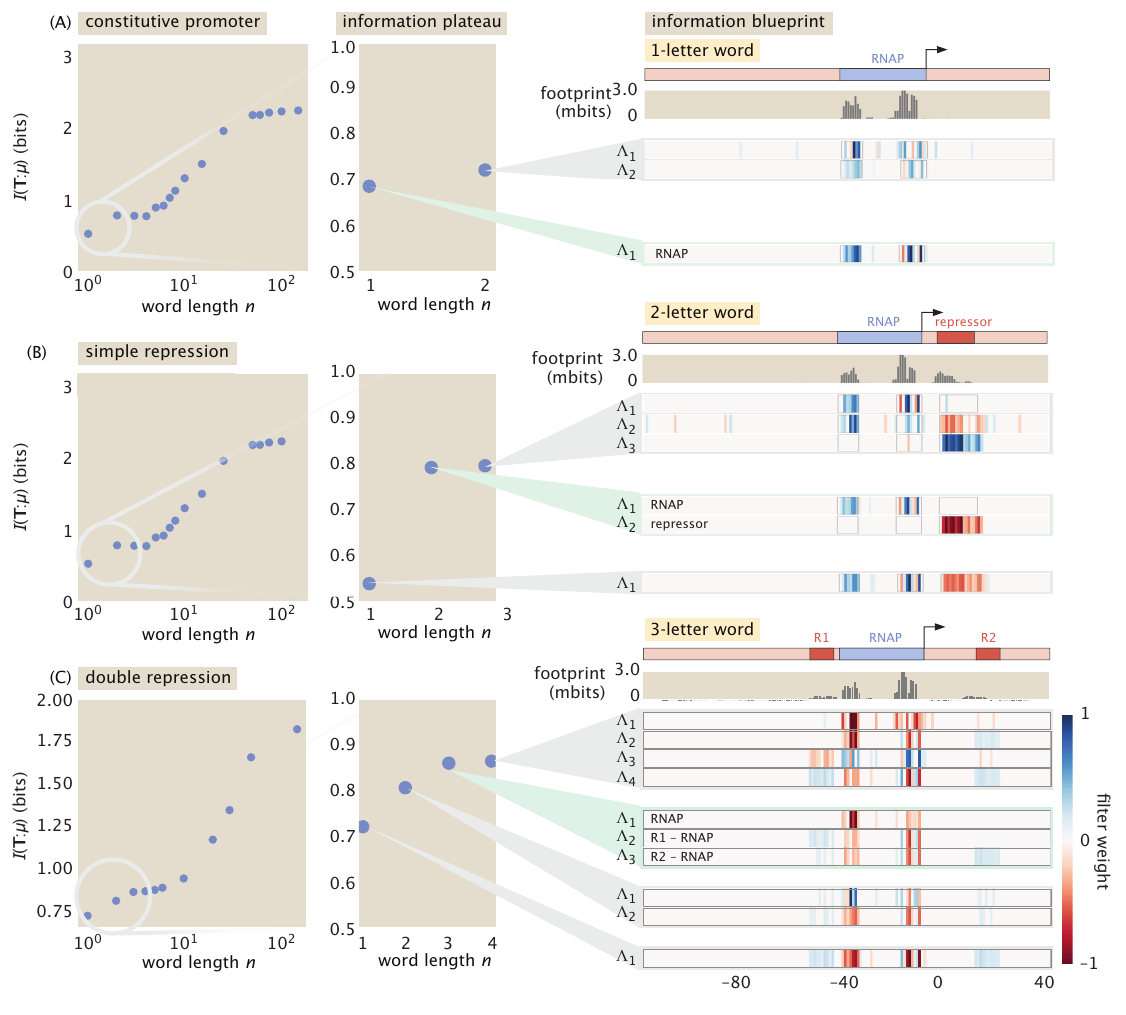}
    \caption{\textbf{Resolving different TF binding sites by tuning the compression rate.}
    Left panels: The mutual information between the sequence and the compressed word is plotted against the word length $n$, defining the rate of compression.
    Here, $n$ corresponds to the number of binary variables $T_\nu$ used to describe the promoter, and is much smaller than the genomic sequence length $N$.
    A small $n$ forces high compression (squeezing features together), while a larger $n$ relaxes the restriction, allowing distinct binding sites to be resolved.
    At small $n$, discrete jumps mark the resolution of each binding site; the curve then reaches a \emph{plateau} once all sites are resolved, whose onset signals the number of regulatory elements. Beyond the plateau, at large enough $n$ approaching the full sequence length, individual base pairs start to be resolved into single-position filters, reflecting a separation of scales between binding-site-level and nucleotide-level structure.
    (A)~Constitutive promoter. A single RNAP site is the only regulatory element. The mutual information reaches a first plateau at $n = 1$, and the single learned filter localizes on the RNAP binding site. Notably, the two domains that belong to the same RNAP site are not split into two separate filters as the regulatory-level information saturates at $n=1$.
    (B)~Simple repression. An RNAP site and a repressor site. At $n = 1$, both sites are conflated into a single filter, with the two regions have couplings with opposite signs reflecting their opposing effects on expression. The information increases upon adding a second hyperletter, and at $n = 2$ the two filters cleanly separate RNAP and repressor. The curve saturates at $n = 2$.
    (C)~Double repression. Two repressor sites flanking an RNAP site. The information plateau is reached at $n = 3$, correctly indicating three distinct regulatory elements. At $n = 3$, each filter localizes on a separate binding site, with both repressor filters showing opposite sign to RNAP.
    }
    \label{fig:wordlength}
\end{figure*}

\subsection{Determining the number of hyperletters and resolution control}

To determine the number of relevant regulatory features, \emph{e.g.}, the number of RNAP and TF binding sites, we systematically increase the word length $n$ until the mutual information $I(\bm{T} : \mu)$ saturates, as shown in Fig.~\ref{fig:wordlength}. As we now proceed to show in detail, this criterion leverages the phase transitions inherent to the information bottleneck~\cite{infbottle1, wu2020phase, parker2022symmetry, gedeon2012mathematical}: discrete jumps in information correspond to the resolution of distinct regulatory elements.

We illustrate this progression through three architectures of increasing complexity. 
As a first testbed, we generate synthetic MPRA libraries from thermodynamic models of transcriptional regulation. 
As explained in more detail in Sec.~\ref{sec:synthetic_validation}, for each promoter architecture, we specify the positions and energetics of the binding sites, draw random point mutations across the promoter, and compute the resulting expression level from the equilibrium occupancy of RNAP~\cite{Pan2024}. This provides ground-truth knowledge of the regulatory elements, against which we can verify that the algorithm recovers the correct number and placement of binding sites.

Figure~\ref{fig:wordlength}~(A) shows the simplest case---a constitutive promoter with only an RNAP binding site. A single hyperletter ($n=1$) already captures the majority of the regulatory information, despite the fact that the RNAP site comprises two spatially separated domains (the $-10$ and $-35$ regions). The two 
domains are not split into separate filters because, at the regulatory level, they function as a single unit. In panel~(B), a repressor site is added alongside RNAP. At $n=1$ the bottleneck is too tight and both sites are conflated into a single filter, with opposite-sign couplings reflecting their opposing effects on expression. Stepping to $n=2$ cleanly separates the RNAP binding site from the repressor binding site. Panel~(C) extends this to double repression, where two independent repressor sites flank RNAP. The information plateau is reached at $n=3$, correctly indicating three distinct regulatory elements, each resolved by its own filter.

In all three cases, the $I(\bm{T}:\mu)$ curve exhibits a discrete staircase of jumps culminating in a stable plateau once all binding sites have been resolved. Beyond this plateau, at large enough $n$ approaching the full sequence length, individual base pairs start to be resolved into their own filters, as expected from the per-position structure of the underlying energy matrix. This plateau therefore reflects a separation of scales between binding-site-level and nucleotide-level structure. For smaller mutagenesis libraries there is an apparent rise at large $n>10$, which is a finite-sample bias of the neural critic estimator of mutual information and is suppressed as the library grows.

We stop at the plateau to precisely target the coarse-grained genomic structure at the scale of binding sites (typically of order 20\,bp), whose collective state determines the bulk of expression---i.e., whether the gene is switched on or off.  This is a choice that stems from our own scientific goals. Pushing $n$ further would resolve individual base pairs and capture fine quantitative variation in expression, but it would also inflate the dimensionality of the representation without adding the regulatory-level structure we seek. The staircase therefore provides a principled way to perform dimensionality reduction of genomic data and to count the number of regulatory elements.

\subsection{\label{sec:bio_prior}Incorporating biological priors for enhanced accuracy and interpretability}

While we have seen in Fig.~\ref{fig:wordlength} that the linear filters $\Lambda_{\nu i}$ successfully localize on binding sites in synthetic MPRA datasets, we can improve robustness and interpretability by incorporating biological priors as inductive biases. 
For microbial organisms, TF binding sites are typically 15--25\,bp long (and typically shorter for mammalian cells)---a length scale set by the physical size of DNA-binding domains. Rather than rediscovering this from data, we can build it into the filter parameterization.

A natural first approach is to constrain each filter to a Gaussian envelope with a learnable center, width, and scalar amplitude, as illustrated in the Methods Fig.~\ref{fig:sliding_filters}. 
This drastically reduces the number of free parameters---from $N$ (the length of the sequence) weights per filter to just three---which proves essential when working with noisy experimental data where unconstrained filters are heavily overparameterized relative to the information signal. However, scalar-amplitude Gaussians impose a strong assumption: they enforce uniform weight within each binding site, preventing the filter from learning position-dependent features such as the internal structure of a TF--DNA contact.

To retain the regularizing benefit of localization while preserving the expressiveness of a freely optimized filter, we introduce an \emph{envelope parameterization}. 
Each filter $\Lambda_{\nu i}$ is written as the element-wise product of an envelope function and unconstrained weights:
\begin{equation}\label{eq:envelope_filter}
    \Lambda_{\nu i} = \mathcal{W}_{\nu i}\, \lambda_{\nu i}\,,
\end{equation}
where $\mathcal{W}_{\nu i} \in [0,1]$ is a smooth envelope that localizes the filter to a contiguous region of the sequence, and $\lambda_{\nu i}$ are freely optimized weights at each sequence position $i$. The envelope softly suppresses the filter outside a window of biologically plausible width, while the weights within that region are free to capture the full positional structure of the binding interaction. A natural choice is a Gaussian envelope,
\begin{equation}\label{eq:envelope_gaussian}
    \mathcal{W}_{\nu i}^{\rm gaussian} = \exp\!\left[-\frac{(i - c_\nu)^2}{2w_\nu^2}\right]\,,
\end{equation}
with learnable center $c_\nu$ and width $w_\nu$. Alternatively, one can use a soft rectangular window constructed from a product of sigmoids,
\begin{equation}\label{eq:envelope_rect}
    \mathcal{W}_{\nu i}^{\rm rectangle}(\tau) = \sigma_\tau\!\left(i - c_\nu + \tfrac{w_\nu}{2}\right)\,\sigma_\tau\!\left(c_\nu + \tfrac{w_\nu}{2} - i\right)\,,
\end{equation}
where $\sigma_\tau(x) = (1 + e^{-x/\tau})^{-1}$ and $\tau$ controls the edge sharpness. The Gaussian envelope provides smoother localization, while the sigmoid product produces a flatter passband that weights all positions within the window more uniformly. In both cases, the width $w_\nu$ is biased to be $\sim 15-25$\,bp to reflect the typical binding site length but it is in principle learnable, allowing the model to adapt to sites of different sizes.

Since each filter carries a single envelope, the number of filters $n$  corresponds to the number of putative binding sites, making the model output interpretable: each filter reports one site's location (via $c_\nu$), width (via $w_\nu$), and detailed contact structure (via $\lambda_{\nu i}$). 
The envelope ansatz recovers more mutual information than the pure Gaussian while remaining sparse, and the learned envelope parameters can be directly read off as binding site positions and widths. 
 Moreover, the resulting optimal envelopes for each hyperletter $T_\nu$ can be used to directly filter out the underlying sequence of the $\nu$th predicted binding site at a base-pair resolution.

Since promoter sequences are discrete, we optimize the envelope centers $c_\nu$ using a discrete update rule rather than standard gradient descent. At each step, we compute the sign of the gradient $\partial I / \partial c_\nu$ and accumulate these signs over a window of training steps. The center is then moved by one base pair in the majority direction only if the accumulated vote exceeds a threshold, ensuring that centers lock onto optimal positions rather than jittering due to gradient noise.

A few practical considerations are worth noting.
The resolution of the blueprint—how many distinct regulatory elements can be resolved—depends on the empirical data distribution, which here is determined by the sampling depth of the mutation library.
Estimating $I(\bm{T}:\mu)$ reliably requires sufficient sequence coverage; sparse or biased libraries will limit the complexity of architectures that can be recovered.
Moreover, in this work, we treat each promoter in isolation to validate the conceptual foundations, though the coarse-graining framework naturally extends to genome-wide analysis where shared filter weights could reveal regulatory motifs common to multiple genes.
To establish the method's capabilities, we first evaluate it under ideal conditions—synthetic data with known ground truth and sufficient sampling—before turning to experimental data where these assumptions may not fully hold.

\begin{figure*}[hbt!]
    \centering
    \includegraphics[]{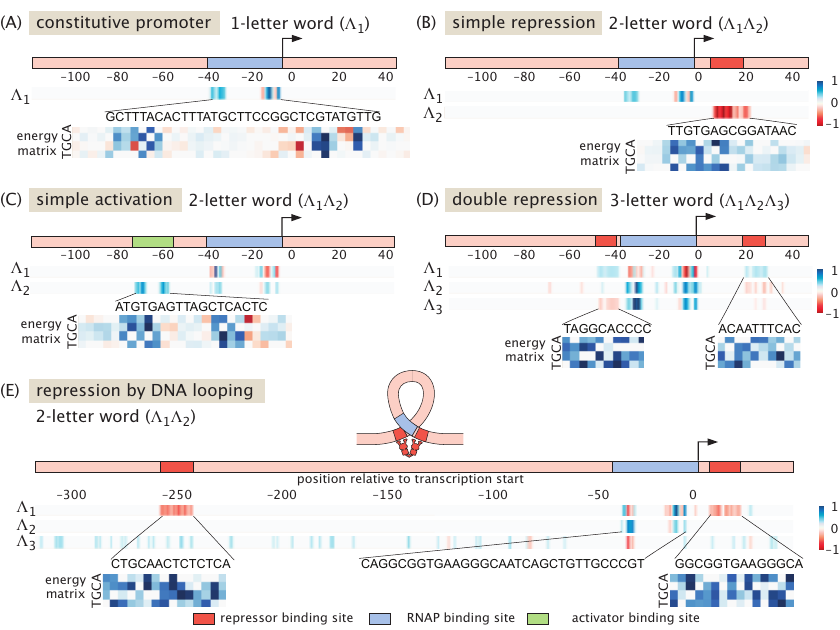}
    \caption{\textbf{Extracting binding sites of synthetic promoters.} 
    Each panel shows a different regulatory architecture. The top bars depict the ground truth locations of binding sites, while the rows below show the optimized compression filters $\Lambda_\nu$. The color map indicates the weight values: blue represents positive weights and red represents negative weights. The relative signs encode the regulatory logic: repressors (red) appear with weights opposite to RNAP (blue), whereas activators (blue) share the same sign. The experimentally determined energy matrices for each binding site are shown alongside for comparison.
    (A) Constitutive promoter. A single filter $\Lambda_1$ recovers the bipartite RNAP binding site.
    (B) Simple repression. Two filters resolve a distinct binding site each: $\Lambda_1$ captures RNAP (positive) and $\Lambda_2$ captures the repressor (negative).
    (C) Simple activation. $\Lambda_2$ identifies the activator site. Its positive sign (concordant with RNAP) correctly identifies the interaction as activation.
    (D) Double repression. With $n=3$, two independent repressor sites are separately coupled by two filters ($\Lambda_2, \Lambda_3$) alongside RNAP.
    (E) Repression by DNA looping. A single filter component ($\Lambda_2$) captures both distal operator sites simultaneously. This reflects the strong mechanical cooperativity due to a loop formation in the polymer, where the two sites function as a single regulatory unit.
    }
    \label{fig:synthetic} 
\end{figure*}

\section{Advantages and Validation of Information Blueprints 
}\label{sec:synthetic_validation}

In this section, we validate our information-blueprint method and compare to alternative approaches such as those based on the information-footprint. 

It is very important that we have some gold standard case studies that allow us to test our systematic coarse-graining approach.  Perhaps the most reliable approach is to first construct synthetic MPRA datasets for which the ground truth is known by using well established thermodynamic models of transcription. In these models, the level of expression 
is proportional to the probability that the promoter of interest is occupied by polymerase~\cite{Shea1985, Buchler2003a, Vilar2003b, bintu2005transcriptional, bintu2005transcriptional2, Sherman2012}.  Initial models of this kind attributed a binding energy for RNAP and transcription factors as a single, sequence-independent parameter. However, in additive models of the binding energy using an ``energy matrix’’ there is a `5 `       34 direct mapping between sequence and binding energy~\cite{Kinney2010Final,Brewster2012,Barnes2019,Pan2024} which allows us to compute the binding energy for any promoter sequence. Using such models~\cite{Pan2024}, we generate synthetic MPRA data for constitutive promoters, the simple repression motif, double repression, simple activation and the case of DNA looping.
The informational blueprints $\Lambda$ are optimized on these synthetic sequence-expression pairs for which we know the ground truth and hence we can verify that the known binding sites are recovered as shown in
Figure~\ref{fig:synthetic}.

With these synthetic datasets in hand, we proceed to consider common regulatory architectures~\cite{Pan2024} and concretely demonstrate four capabilities of informational blueprints: (a) grouping correlated nucleotide positions at distant positions into functional binding sites; (b) inferring the regulatory sign of the transcription factor associated with each binding site;
(c) inferring the number of distinct regulatory elements in a promoter under certain environmental conditions,
even if the binding sites are overlapping;
(d) identifying the logic gates that determine how multiple transcription factors regulate a gene.

\subsection{Grouping correlated positions into binding sites}
A key feature of informational blueprints is their ability to detect and group correlated nucleotide positions into functional binding sites, even when those positions are dispersed across the sequence. This follows from the fact that the blueprints are obtained by optimisation of the \textit{global} information content rather than by scanning individual positions.

To a physicist's eye, there is a useful analogy with statistical mechanics: 
information footprints are essentially mean-field descriptions in the sense that they attempt to reduce the promoter into an effective single-site contributions $I(B_i:\mu)$. This is analogous to representing a magnetic bar comprised of many interacting spins, each oriented either $\uparrow$ (here, hyperletter $T=\texttt{1}$) or $\downarrow$ ($T=\texttt{0}$),  as a single spin subject to a mean field that attempts to capture its interaction with all other spins.
By contrast, information blueprints explicitly take into account many-site correlations and mimic tools like the renormalization group designed to capture cooperative behaviour in magnets that arises near phase transitions and cannot be captured by mean-field approaches.

The algorithmic ability to group many correlated base-pairs into a single binding site follows from the sparsity of the weights in the information blueprint. For a given hyperletter $\nu$, the weights $\Lambda_{\nu i}$ end up being approximately zero everywhere except at the positions $i$ that are inside a specific binding site, as shown in Figures~\ref{fig:wordlength} and \ref{fig:synthetic}.
The optimization drives the entire weight vector $\bm{\Lambda}_\nu = (\Lambda_{\nu 1}, \dots, \Lambda_{\nu N})$ for each filter to sparsely concentrate on specific promoter regions.

This feature can be demonstrated already on the simplest case with a synthetic promoter with only an RNAP binding site and no TFs, as shown in Fig.~\ref{fig:synthetic}(A).
In this case, only one hyperletter encoding a single filter ($n=1$) suffices: the optimized filter $\Lambda$ localizes on the RNAP binding site, with the characteristic bipartite structure reflecting the $-10$ and $-35$ boxes.
Positions outside the binding site receive negligible weight (smaller than 10\% of the max value), confirming that the information bottleneck forces the compression to focus on regulatory elements.
This also shows that, without providing any prior biological knowledge, blueprints discover the typical size of regulatory proteins.
Adding more filters ($n>1$) gains no information by separating the two domains of the RNAP binding, as shown in Fig.~\ref{fig:wordlength}(A). This observation confirms that the delocalised positions are indeed associated with a \emph{single} binding site.

\subsection{Inferring regulatory sign}
The blueprint weights $\Lambda_{\nu i}$ naturally carry a sign, encoding the direction of the regulatory interaction directly within the compression---activators and repressors are distinguished automatically, without supplementary statistics.

Although the overall sign of any single filter is arbitrary (since the flipping of $\bm{\Lambda}_\nu \to -\bm{\Lambda}_\nu$ and $T_\nu \to \texttt{NOT}(T_\nu)$ leaves the mutual information unchanged), the \emph{relative} signs between different weights have biological meaning.
Regulatory elements that cooperate with RNAP, \emph{e.g.} activators, share the same sign, since mutations in either site disrupt expression.
In contrast, repressors appear with the \emph{opposite} relative sign: a mutation that breaks a repressor site lifts inhibition and increases expression, opposing the effect of breaking the RNAP site.

To demonstrate this feature, we consider regulatory architectures with a single TF binding site in addition to RNAP, as shown in Figs.~\ref{fig:synthetic}(B) and (C). With $n=2$ hyperletters, the compression separates RNAP from the TF site, assigning each to a distinct filter component. The sign of the weights distinguishes regulatory function: for repression, the TF filter has the opposite sign to the weights at the RNAP site, reflecting that mutations in the repressor site increase expression while mutations in the RNAP site decrease it. This is in contrast with a promoter with an activator site, shown in Fig.~\ref{fig:synthetic}(C), where both filters share the same sign—disrupting either site reduces expression. This sign structure is not imposed but emerges from the optimization.

We note that this is a feature absent from information footprints, which yield only a positive number $I(b_i:\mu)$ at each position and therefore require supplementary statistics to determine the sign of the regulatory effect \footnote{This is analogous to mean-field treatments of spin systems, where only the absolute value of the magnetisation is accessible, losing information about whether the magnetization is oriented $\uparrow$ or $\downarrow$.}.

\subsection{Counting the number of binding sites}
Different information blueprint filters $\bm{\Lambda}_\nu$ tend to localize on distinct regulatory elements, providing a direct and systematic way to count the number of binding sites in a promoter. This counting procedure was explicitly illustrated in Fig.~\ref{fig:wordlength} and validated with synthetic data. In Fig.~\ref{fig:synthetic}(A) we consider the case of a single RNAP site (one-hyperletter word, $n=1$) and in Figs.~\ref{fig:synthetic}(B) and (C) the case of RNAP coupled to a repressor or activator site (two hyperletter words, $n=2$).

In particular, the blueprint approach can delineate binding sites even when they spatially overlap.
Appendix Fig.~\ref{fig:overlap} shows a synthetic promoter where the RNAP and repressor sites share common positions.
At $n = 1$, the compression bottleneck forces both regulatory elements into a single conflated filter.
Increasing to $n = 2$ resolves them into two distinct components, each localizing on its respective binding site.

The counting by informational blueprints works even for more complex architectures that require additional hyperletters to resolve distinct regulatory elements.
For example, in Fig.~\ref{fig:synthetic}(D), a promoter with two repressor sites with independent operators, each repressor binds its own site and acts autonomously—either one can silence the gene.
Here, three filters ($n=3$) cleanly separate all three binding sites: RNAP and the two repressors, with both repressor filters showing opposite sign to RNAP.
The independence of the two repressors is reflected in the filter structure: each occupies its own filter.

DNA looping, shown in Fig.~\ref{fig:synthetic}(E) presents a particularly interesting case where there is an interplay between the two tasks of detecting non-local correlations and counting binding sites. Consider, two repressor proteins that form a tetramer and bind two spatially separated operator sites, which forces the intervening DNA to bend into a loop. The two operators do not function independently---they are mechanistically coupled, and both must be occupied for the looped configuration to form and block transcription.
Intuitively, the blueprints combine both distant operators within a \emph{single} filter, reflecting their cooperative function as one non-local regulatory unit. This illustrates a key strength of the global compression scheme: it automatically detects that two distant sequence regions function together as a single regulatory element.

There is an analogue of this automated counting in the physics of phase transitions: it is the automated identificat
ion of the symmetries of the order parameter. More precisely $2^n$ would correspond to the number of elements possessed by the generator of the discrete symmetry group representing the order parameter, e.g., $n=3$ could represent a ($2^3=8$)-fold rotation because each of the $n=3$ letters can attain two values \texttt{1} or \texttt{0}.  

\subsection{Inferring logic gates of transcription factors}

One of the biggest challenges to understanding regulatory architectures is to infer how multiple TF binding sites cooperate to control genes. This is typically attempted by performing (in the lab or in silico) \textit{multiple} experiments sampling a range of copy numbers of TF. By contrast, the information blueprint can reveal cooperative effects involving multiple TF binding sites with a single experiment at a specific TF copy number.

Appendix Fig.~\ref{fig:logic} shows how the blueprints reveal the regulatory logic governing how multiple TFs are integrated, using synthetic data. We consider double repression, where two repressors R$_1$ and R$_2$ control expression. Under \texttt{AND} logic, repression requires both repressors to bind; under \texttt{OR} logic, either suffices. Since our filters couple to mutations that disrupt binding, they effectively compute functions of the variables $r_i$, which is $\texttt{1}$ when the binding site is disrupted and $\texttt{0}$ otherwise.

\begin{figure*}[hbt!]
    \centering
    \includegraphics[width=0.85\textwidth]{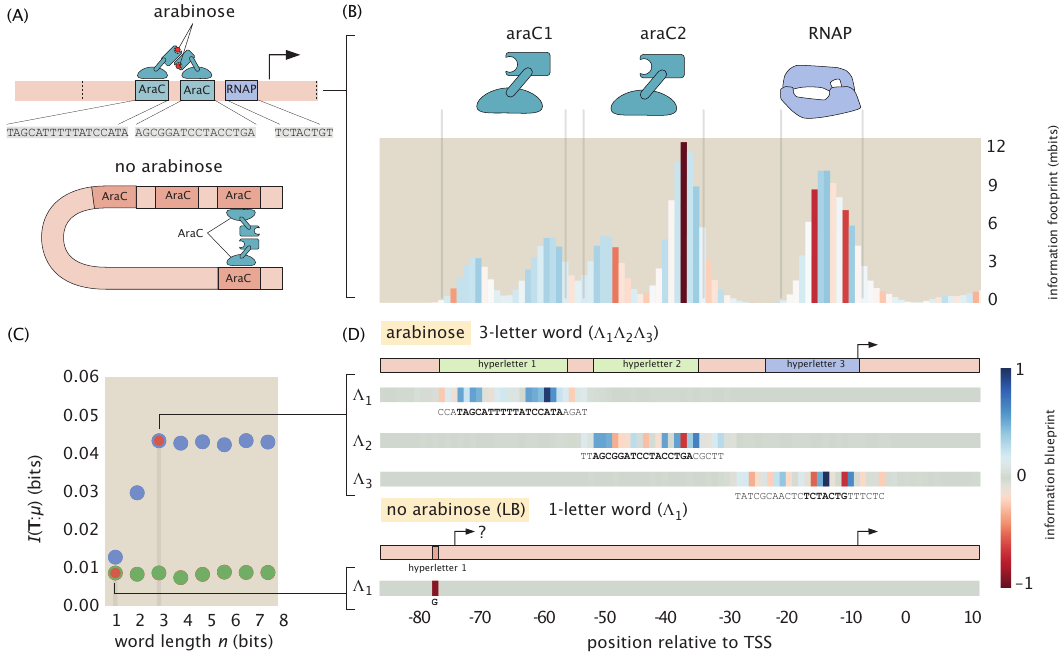}
    \caption{{\bf Recovering regulatory architecture for arabinose operon from MPRA data.} 
    (A)~Schematic of the \emph{araBAD} promoter architecture. When arabinose is present, AraC binds two sites.
    (B)~Information footprint showing mutual information $I(b_i:\mu)$ at each position along the promoter in the presence of arabinose, showing at least five peaks.
    (C)~Compressed information $I(\bm{T};\mu)$ versus word length~$n$ plateaus at $n=3$, correctly predicting three binding sites. 
    In the absence of arabinose (red), the information remains small at all~$n$, consistent with the loss of AraC-mediated activation.
    (D)~In the presence of arabinose, the three optimised filters localise on the known binding sites, and the extracted subsequences match the established binding motifs~\cite{schleif_regulation_2000}. Notably, in the absence of arabinose, a single point mutation in this background can restore non-zero information, indicating the creation of a new transcription start site.
    }
    \label{fig:arabinose}
\end{figure*}

For \texttt{AND}-repression, expression is high whenever at least one site is disrupted ($r_1$ \texttt{OR} $r_2 = \texttt{1}$)—a function computable by a single linear filter $T_0 = r_1 + r_2$. Indeed, we find that one filter suffices, with weights coupling to both repressor sites with the same sign. For OR-repression, expression is high only when \emph{both} sites are disrupted ($r_1$ \texttt{AND} $r_2 = \texttt{1}$)—an \texttt{AND} gate that no single linear filter can compute. Accordingly, the optimization requires two filters, each coupling to one repressor site independently. This two-hyperletter representation uses more states than strictly necessary, but succeeds because it distinguishes all functionally distinct configurations. The number of filters and their weight structure thus encode the regulatory logic: shared weights indicate cooperative regulation, while separate filters indicate independent action.

The ability of informational blueprints to capture cooperative effects involving interacting transcription factors stems directly from its physics-inspired origin. 
Renormalization group approaches in physics can construct, upon coarse-graining (e.g. spins decimation), effective descriptions of complex systems (e.g. a magnet) at each level of a hierarchy of spatial scales. 
Similarly, our information blueprint method can construct,  upon coarse graining the sequence at nucleotide resolution, an effective logic circuit of binding sites at the level of regulatory proteins and potentially even at the level of groups of genes, if the coarse-graining of genomic information could be iterated.

\subsection{Robustness to noise}
In practice, MPRA measurements are affected by experimental noise: stochastic variation in DNA and RNA counts, barcode sampling, and other sources of technical variability.
To assess the robustness of our method, we generated synthetic simple-repression libraries with increasing levels of measurement noise applied to the RNA/DNA ratio.
The SI Figure~\ref{fig:noise_benchmark} summarises the results: the left panels show the captured mutual information $I(\bm{T};\mu)$ as a function of word length $n$ at each noise level, while the right panels display the corresponding learned filters at $n=2$.
As noise increases, the mutual information decreases and the signal to noise ratio degrades, making it more difficult to resolve distinct binding sites. In the free parameterization, the filters become increasingly diffuse and noisy as illustrated in the top row of the SI Figure~\ref{fig:noise_benchmark}, and the information plateau becomes less pronounced, indicating that
additional regulatory elements can no longer be clearly discerned in separate filters.
Crucially, the biologically informed envelope parameterisation (Section~\ref{sec:bio_prior}) extends the resolvable regime: by constraining each filter to a physically plausible width and position, the number of free parameters is drastically reduced, allowing the optimisation to distinguish genuine binding-site signals from noise even when unconstrained filters fail to do so, as shown in the bottom row of the SI Figure~\ref{fig:noise_benchmark}.

\section{Predicting Experimental Binding Sites on \emph{E. coli} Promoters}\label{sec:experimental}

The central goal of our analysis is to use our systematic approach in the context of MPRA experiments that allow us to rigorously and reliably determine the regulatory architectures for all promoters across a genome in a suite of different environmental conditions.  
Unlike bioinformatic approaches that postulate putative binding sites as overrepresented motifs by considering the sequences alone~\cite{bussemaker_building_2000}, our method identifies 
functionally active regulatory elements by quantifying the relevance of mutations for expression---making the prediction 
inherently condition-specific.
To illustrate our progress towards this ambitious goal, we now deploy our method at scale on experimental data using the same pipeline we illustrated with synthetic data.  We begin with the gold standard arabinose operon for which the experimental ground truth is well known (III A). We then deploy informational footprints at scale for the \textit{tisB} promoter across 40 growth conditions (III B). Lastly, we illustrate how to use our approach to discover regulatory architectures in {\it E. coli}. including the case of promoters that have not been previously annotated.

\begin{figure*}[htb!]
    \centering
    \includegraphics[width=0.95\textwidth]{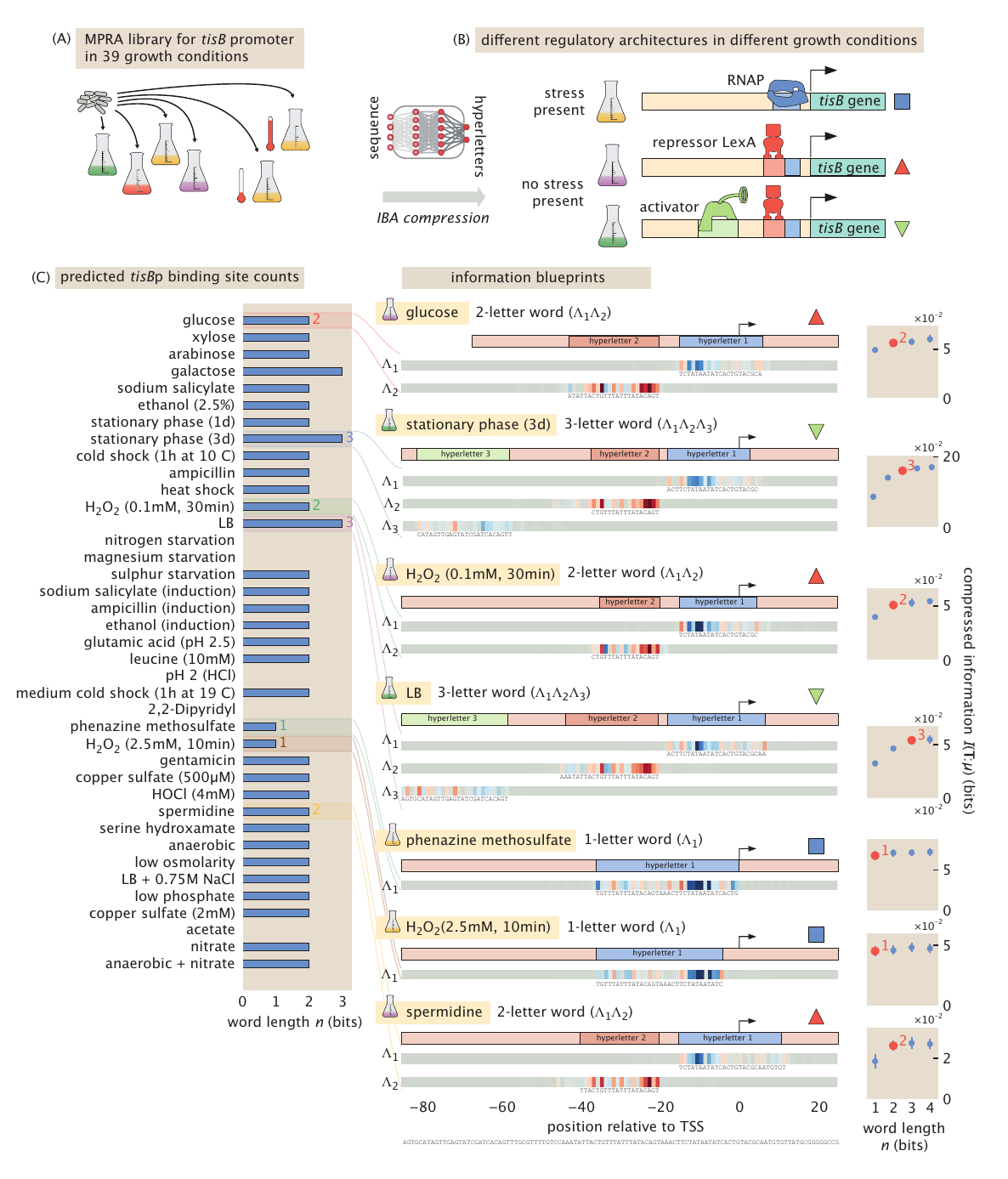}
    \caption{{\bf Deploying the method at scale on MPRA data for the \emph{tisB} promoter across 39 growth conditions.}
(A)~The \emph{tisB} promoter library is assayed by MPRA under 40 distinct growth conditions, each yielding a condition-specific expression profile.
(B)~Schematic illustrating that different growth conditions can activate different regulatory architectures on the same promoter, engaging distinct combinations of binding sites.
(C)~Predicted number of binding sites (word length $n$) for each growth condition. Colored labels mark seven conditions representing the 3 possible regulatory architectures.
(D)~Information blueprints for the representative conditions for $n=1,2,3$ putative binding sites. For each condition, the optimized filter components $\Lambda_i$ are displayed as heatmaps over the promoter sequence, with the corresponding wild-type subsequences shown above. The right column shows the compressed mutual information $I(\boldsymbol{T};\mu)$ as a function of word length~$n$. The red dot indicates the predicted number of binding sites, at which the information gain due to additional hyperletters is less than 10\%. Conditions are ordered from top to bottom following the ranking in~(C), revealing a spectrum from simple single-site regulation (e.g.\ glucose) to multi-site architectures involving up to three predicted binding sites (e.g.\ stationary phase).
}
    \label{fig:deploying_at_scale} 
\end{figure*}

\subsection{Benchmarking on the arabinose operon}\label{sec:arabinose}
Our understanding of gene regulation owes much to early work on the choices that bacteria make about preferred carbon sources.
In the 1960s, the operator-repressor model of Jacob and Monod~\cite{jacob_genetic_1961} was vastly extended through the discovery of gene activation~\cite{englesberg_positive_1965, zubay_mechanism_1970}. One of the central case studies leading to this understanding was the well-studied arabinose operon, shown schematically in Fig.~\ref{fig:arabinose}(A).

Concretely, the promoter of \emph{araBAD} is regulated by AraC, which binds multiple operator sites to activate transcription in the presence of arabinose~\cite{schleif_regulation_2000}.
As shown in Fig.~\ref{fig:arabinose}(B), when arabinose is present, there exist at least five major peaks in the local mutual information.
Recall from Section~\ref{sec:synthetic_validation}, the footprints alone indeed do not provide a systematic way to group these peaks into the larger structures that constitute binding sites.
To see how information blueprints circumvent this ambiguity, we first plot the mutual information of the compressed words $I(\bm{T}:\mu)$ versus number of hyperletters.  We find that the information plateaus at three hyperletters with no improvements deriving from adding extra ones
as shown in Fig.~\ref{fig:arabinose}(C). Our algorithm, suggests that only three binding sites exist.

Once we identified the number of binding sites, we can determine their location using the compression filters associated with each of them. To do that we use the strategy discussed in Section~\ref{sec:bio_prior} comprised of (i) a sliding envelope that searches for an optimal position of the filters and (ii) an adaptive scheme to optimize the trial function within each filter as revealed in Fig.~\ref{fig:sliding_filters}. As shown in Fig.~\ref{fig:arabinose}(D), the three optimised filters localise precisely on the known two AraC binding sites and the RNAP site.

In this figure, we introduce our graphical notation for characterizing the positions of the different hyperletters.  For example, we see that for hyperletter 1, it is associated with an activator binding site.  This implies that $T_1 = 1$ for all mutant versions of this binding site that are sufficiently disruptive to impact expression and $T_1 = 0$ for those mutants which preserve high expression. Further, the colored base pair resolution maps below each schematic show the weights associated with each base in the filter $\Lambda_1$.

In addition we use the optimised envelope functions to select the sequence associated with each binding site and compare it to the known sequences, where we neglect letters at positions near the envelope edges with weights smaller than 15\% of the maximum value in each filter component.
The resulting extracted subsequences match the known motifs shown in Fig.~\ref{fig:arabinose}~(A).

The condition-dependence of regulatory architecture is already visible in this example.
In the absence of arabinose, AraC no longer activates transcription and the compressed information $I(\bm{T};\mu)$ remains constant at all word lengths~$n\geq 1$, as shown by the red markers in Fig.~\ref{fig:arabinose}(C).
Strikingly, the optimal filter couples strongly to only a single isolated position whose mutation restores non-zero information.
This single mutation is thought to create a new transcription start site.
This demonstrates that the method can not only recover regulatory elements but also flag more exotic features such as latent promoter creation by point mutations.

\subsection{Deploying informational blueprint at scale across many growth conditions}

In moving from synthetic systems to real MPRA data, an additional ingredient becomes central: growth conditions.
Indeed, the very essence of gene regulation is to turn genes on or off depending on the environment, ultimately giving rise to different cell states and types.
Transcription factors are only active if they are present, modified, or ligand-bound, and these states are set by the environment~\cite{Roeschinger2026}.
Growth conditions are therefore an essential input to the regulatory logic circuit itself: different conditions selectively activate different branches of the underlying promoter architecture, so that the regulatory blueprint---including the number and identity of operative binding sites---is inherently condition-dependent.  Indeed, we advocate for the idea that in the future, databases will explicitly acknowledge how regulatory architecture depends upon conditions, revealing different constellations of binding sites depending upon those conditions.

A key advantage of the information-blueprint framework is that it can be applied systematically across a large number of conditions for the same promoter, mapping out how the operative regulatory architecture changes with the environment.
We demonstrate this on the \emph{tisB} promoter~\cite{dorr2010ciprofloxacin, su2022tisb}, for which MPRA data are available in 39 distinct growth conditions~\cite{Roeschinger2026}.

For each condition we run the full compression pipeline from section~\ref{sec:bio_prior} and record both the predicted number of binding sites and the corresponding information blueprints. 
We determined plateaus by checking if the information gain due to additional hyperletters is less than 10\%.
The results, summarised in Fig.~\ref{fig:deploying_at_scale}, reveal a rich condition-dependent spectrum: from simple single-site architectures (e.g.\ hydrogen peroxide) to multi-site regulatory logic involving up to three predicted binding sites (e.g.\ stationary phase).
Across most conditions for this promoter we find a repressor site between positions $-35$ and $-15$ relative to the TSS.

In Fig.~\ref{fig:discovery} we show the information blueprints for four additional promoters.
The compression yields concrete predictions: putative binding sites, their approximate positions and---via the relative signs of the filter weights---the likely regulatory role (activator vs.\ repressor) of each site.

\begin{figure*}[hbt!]
    \centering
    \includegraphics[width=0.95\textwidth]{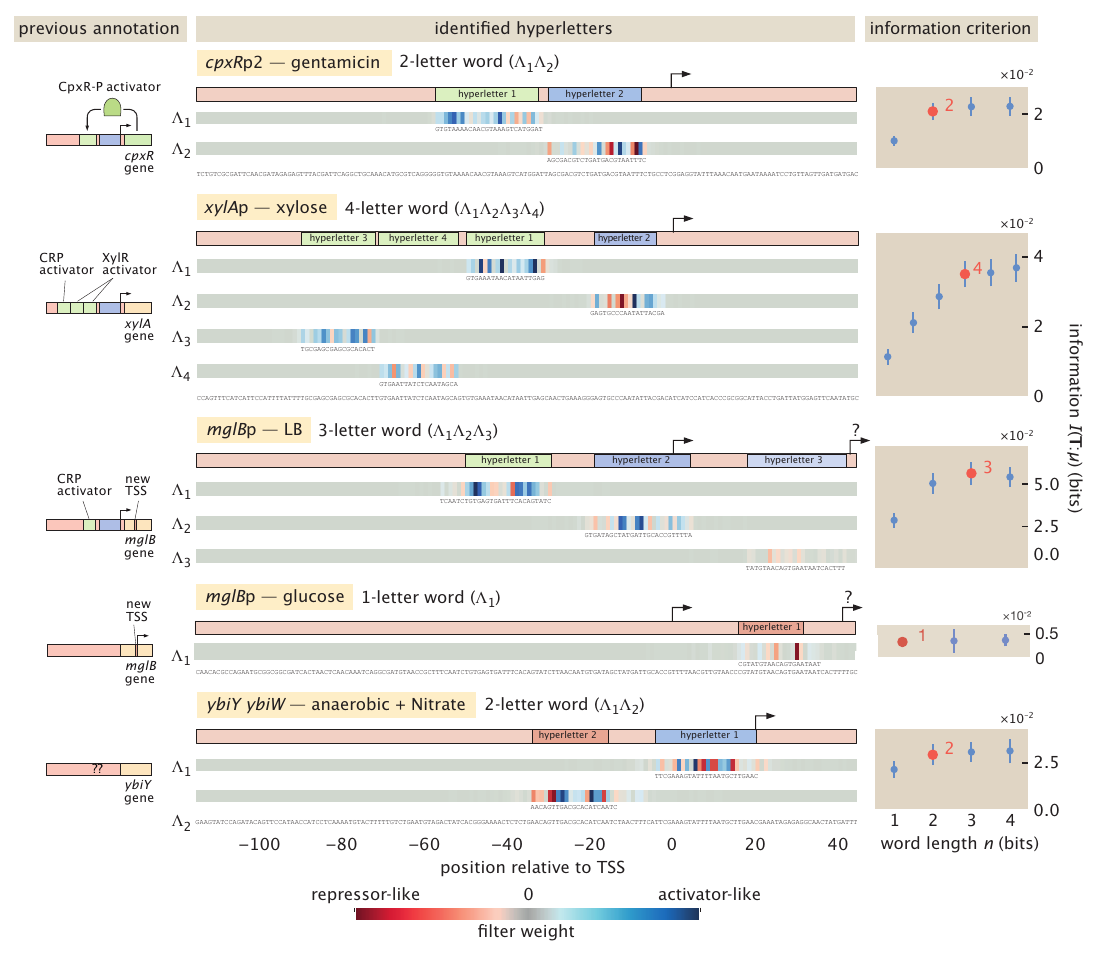}
    \caption{{\bf Discovering regulatory architectures in \emph{E.\ coli}.}
Information blueprints for four promoters that were previously annotated partially, illustrating the generality of the method. Each panel shows the identified filter components $\Lambda_i$ (heatmaps with wild-type subsequences) and the compressed information curve $I(\boldsymbol{T};\mu)$ versus word length~$n$ for a representative growth condition.
\emph{cpxR}p2 in gentamicin (2~binding sites),
\emph{xylA}p in xylose (4~binding sites),
\emph{mglB}p in LB and glucose (2~binding sites and a new TSS), and
the predicted \emph{ybiY} promoter under anaerobic conditions with nitrate (2~binding sites).
The diversity of filter structures across promoters and conditions demonstrates that the method can recover condition-specific regulatory logic from MPRA data without prior knowledge of transcription factor binding motifs.
}
\label{fig:discovery}
\end{figure*}

Sequence-similarity methods have annotated this region as a LexA binding site---but such annotations are static: they cannot distinguish conditions where the site is 
operative from those where it is not. Our blueprints can. Under hydrogen peroxide (H$_2$O$_2$) and phenazine methosulfate stress, the repressor signal vanishes. This is consistent with the known biology: \emph{tisB} is a toxin gene repressed by LexA~\cite{d1985sos, roth2022transcriptomic, little1991mechanism, giese2008reca}, and DNA damage triggers the SOS response, relieving this repression. The disappearance of the repressor signal in precisely these stress conditions confirms that the blueprint captures condition-specific regulatory logic. This analysis illustrates how the method scales to database-level deployment: the same automated pipeline processes all 39 datasets, returning a condition-by-architecture map of which regulatory elements are active under which environmental signals.

\subsection{Discovering regulatory architectures in \emph{E.\ coli}.}

We now deploy our information blueprint approach to discover regulatory architectures in {\it E. coli} starting with known cases and culminating in  promoters that have not been previously annotated, see Fig. \ref{fig:discovery}.

For the \emph{cpxR}p2 promoter under gentamicin stress, the blueprint identifies two binding sites.
This is consistent with the known positive autoregulation of the CpxRA two-component system: CpxR-P binds a consensus motif in its own promoter region, amplifying the envelope stress response~\cite{dewulf1999}.
The two predicted sites correspond to the CpxR-P autoregulatory site and the RNAP site, respectively.

For \emph{xylA}p in xylose, the compression resolves four nearly adjacent binding sites.
The xylose utilisation operons are known to be co-regulated by the activator XylR, which binds two direct-repeat sites upstream of the $-35$ region, and by CRP, whose binding site flanks the XylR operator~\cite{song1997}.
The four-site architecture recovered by the blueprint is compatible with these known sites; the first and fourth filters likely capture the two XylR sites, while the second and third filters correspond to the RNAP and CRP sites, respectively.

We then look at the galactose/methyl-galactoside transport operon, \emph{mglB}p. For this gene we find three sites when the carbon source is LB or xylose.
This promoter is known to be activated by CRP and repressed by the galactose-responsive regulators binding the \emph{mgl} operator~\cite{weickert1993, KRISHNA2009671}.
While the first two filters of the blueprint capture the known CRP and RNAP sites, intriguingly, the third filter localizes to a distinct highly local region which has been identified as a new transcriptional start site, as a mutation at the +30 position leads to a consensus -10 sequence~\cite{Roeschinger2026}, suggesting that this site may play a regulatory role by modulating the balance between two alternative TSSs.
In contrast, when the carbon source is glucose, we found a single-letter word that locally couples to this predicted new TSS at +30 position.

Finally, we investigate the putative promoter region of \emph{ybiY}, for which no regulatory annotation is currently available.
Under anaerobic conditions, our algorithm yields two filters.
The second filter has activator-like couplings in the -30 -- -15 region, which may be part of an RNAP site.
More surprisingly, the first filter exhibits strong couplings in the $0$ to $+20$ region---downstream of the annotated transcription start site and inside the gene body. One possible explanation is that the predicted TSS for \emph{ybiY} was inaccurately assigned; alternatively, a genuine regulatory element may reside within the transcript leader.

These results demonstrate how our pipeline can produce concrete hypotheses for the environment-dependent regulatory architectures of promoters, which can be experimentally tested in the wet lab.
For example, one strategy, as may be suggested by the structure of the compression itself, is to introduce targeted mutations in the regions identified by the filters.
Because each filter partitions sequences into those with intact vs.\ disrupted binding sites, mutating the predicted region effectively constructs the conditional distribution $P(\mu| T_\nu = 0)$ experimentally.
Comparing expression in such a targeted library against the wild-type background would directly probe whether the predicted site carries regulatory information.

\section{Conclusions and Outlook}

The compression framework presented here offers a principled way to extract regulatory architectures from high-throughput sequence-expression data  without prior assumptions about motifs, TF identities or interactions.
Deployed on synthetic and experimental data, our method recovers known binding sites and regulatory logic in well-characterized promoters, while generating testable predictions in unannotated y-ome promoters.
Since the method optimizes a global information-theoretic objective, it naturally captures cooperative interactions between binding sites, including non-local ones that induce structural changes on the DNA itself.
This provides a systematic way to connect MPRA data to underlying biophysical models of gene regulation.

Beyond annotation, our method can provide insights into the connection between DNA sequence and evolution. 
Though of course much evolutionary change focuses on the protein coding sequences of the genome, there are also scores of examples where instead, it is the regulatory part of the genome that is mutated.  
Further, this approach can also provide actionable guidance for promoter engineering: the filter weights indicate which positions to target for tuning expression levels, and the sign structure reveals whether mutations will increase or decrease output. This could accelerate the design of synthetic regulatory circuits with predictable behavior.

Positions identified by the information blueprint—those with high filter weights—are predicted to be under selective constraint: mutations here significantly alter expression and thus may be likely purged by selection. This connects our information-theoretic approach to evolutionary signatures of functional constraints, offering a complementary lens to phylogenetic footprinting methods that rely on cross-species conservation.

Looking ahead, the compression of individual promoters here can be viewed as the first step of an iterative coarse-graining procedure. 
The full many-body problem of gene regulation—how genes collectively give rise to cellular phenotypes—requires understanding interactions between genes, not just within them.
Just as real-space renormalisation group methods decimate microscopic degrees of freedom to reveal effective interactions at larger scales, one could imagine applying this framework genome-wide—compressing nucleotide sequences into binding configurations, then binding configurations into gene-gene interactions, and ultimately to the low-dimensional cell states that emerge from regulatory network dynamics.

\section{Materials and Methods}

\subsection{Optimal lossy compression framework}\label{sec:framework}

\subsubsection{Optimal compression filters localise on binding sites}\label{sec:covariance}
 
Here we sketch a simple analytical argument explains why the optimal compression filters $\bm{\Lambda}_\nu$ localise on binding sites. In the linear regime ($\sigma\approx\mathrm{id}$), the hyperletter $T^{(m)}=\sum_i\Lambda_i\,B_i^{(m)}$ is a scalar projection, and the IB objective $\max_{\bm{\Lambda}}I(T;\mu)$ reduces to maximising the squared correlation between $T$ and~$\mu$. Writing $\mathrm{Cov}(T,\mu)=\sum_i\Lambda_i\,C_i$ with
\begin{equation}\label{eq:cov_binary}
  C_i \;=\; \mathrm{Cov}\bigl(B_i,\;\mu\bigr) \;=\; \bigl\langle B_i\,\mu\bigr\rangle - \bigl\langle B_i\bigr\rangle\bigl\langle\mu\bigr\rangle
\end{equation}
and $\mathrm{Var}(T)=\bm{\Lambda}^{\!\top}\Sigma_{BB}\,\bm{\Lambda}$ where $(\Sigma_{BB})_{ij}=\mathrm{Cov}(B_i,B_j)$, the optimum of $(\bm{\Lambda}\cdot\bm{C})^2/(\bm{\Lambda}^{\!\top}\Sigma_{BB}\,\bm{\Lambda})$ is $\bm{\Lambda}\propto\Sigma_{BB}^{-1}\,\bm{C}$. In a randomly mutagenised library each position is mutated independently with probability~$p_{\mathrm{mut}}$, so $\Sigma_{BB}$ is approximately diagonal and the optimal filter simplifies to $\Lambda_i\propto C_i$.
 
Because expression depends on sequence almost exclusively through the binding energies of the regulatory proteins, $\mu = \mu(E_1,\dots,E_{n_{\mathrm{TF}}})$, and each energy $E_a = \sum_{i\in S_a}\varepsilon_{i,s_i}^{(a)}$ is affected only by mutations within its own site~$S_a$, the chain rule gives
\begin{equation}\label{eq:cov_decomp}
  C_i \;=\; \sum_{a=1}^{n_{\mathrm{TF}}} \chi_a\;\delta\bar{E}_i^{(a)}\;\mathbf{1}_{i\in S_a},
\end{equation}
where $\chi_a = \langle \partial\mu/\partial E_a \rangle_{\mathrm{eff}}$ is the thermodynamic susceptibility of expression to the binding energy of TF~$a$, and $\delta\bar{E}_i^{(a)}$ is the binding energy shift caused by a mutation at position~$i$, averaged over the three possible substitutions of a base. Equation~\eqref{eq:cov_decomp} is non-zero only at positions that belong to a binding site, explaining filter localisation. It also shows that cooperative interactions---which correlate the susceptibilities $\chi_a$ across the mutant ensemble---can cause a single filter to span multiple sites (see Appendix~\ref{app:constitutive} for the constitutive promoter and Appendix~\ref{app:dna_looping} for DNA looping).

\subsubsection{Determining the rate of compression and information bottleneck phase transitions}\label{sec:phase_transitions}
In the information bottleneck framework, the most relevant features are targeted by enforcing a sufficiently high rate of compression in Eq.~\ref{eq:objective}.
Here we fix the rate of compression by choosing a certain number of filters since the information capacity of the compressed variable is bound by its number of components. 

Crucially, as the rate of compression is reduced, \emph{e.g.} by adding more filters, the mutual information does not increase smoothly, but it rather undergoes a sequence of sharp transitions, each following a plateau of saturation~\cite{NIPS2002_ccbd8ca9}. 
These transitions are associated with a discrete hierarchy of relevant features of the data.

Suppose that the base-sequence data $\bm{B}$ can be decomposed into a set of features $\{f_1, f_2, \ldots \}$ with a decreasing order of relevance for the expression level $\mu$\footnote{Our hypothesis is that the most relevant feature $f_1$ is a representation of the binding sites of the TFs.}.
For the highest rate of compression, i.e. for a 1-bit variable $\bm{T}$, the filter $\Lambda$ will only target the most relevant feature $f_1$ partially.
Upon adding further components into $\bm{T}$, instead of mixing in sub-relevant features $f_2, f_3, \ldots$, the optimal compression instead utilises these components to express the more fine-grained aspects of the same most relevant feature $f_1$. 
When $f_1$ is finally expressed perfectly, the mutual information then saturates until the number of filters $n$ becomes sufficiently large to start expressing the next relevant feature $f_2$.

The sufficient number of filters in the optimal compression scheme can therefore be determined systematically by monitoring the first saturation of mutual information as a function of the number of filters $n$. 
In the case of the simple repression architecture, we see that the mutual information saturates at two filters, matching the number of distinct TFs in the system.

\begin{figure*}[hbt!]
    \centering
    \includegraphics[width=0.9\textwidth]{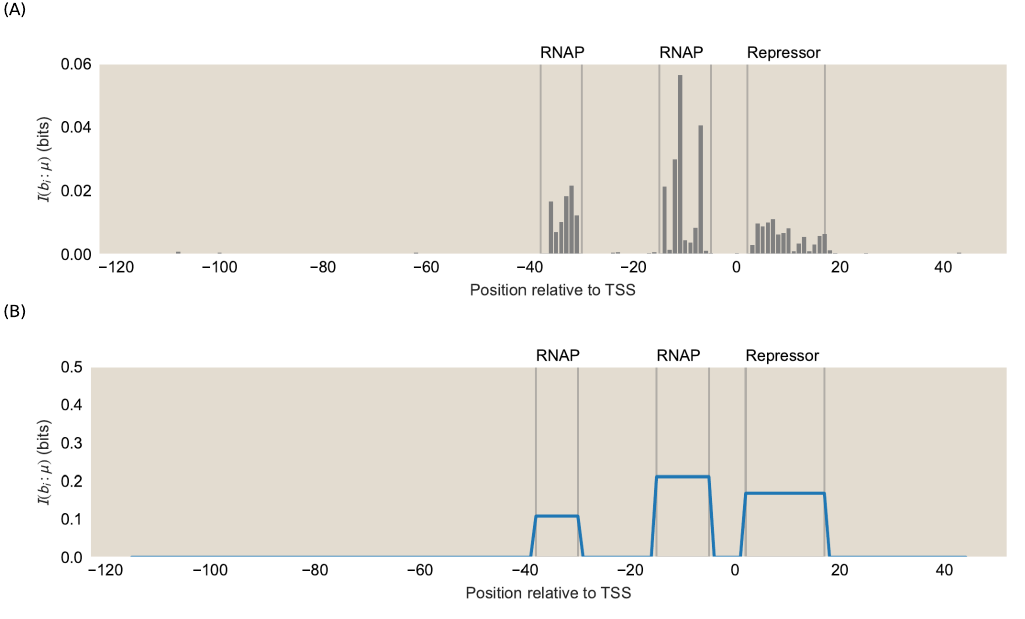}
    \caption{\textbf{Information footprints for the simple repression architecture (synthetic data).}
        (A)~Per-site information footprint $I(B_i:\mu)$ computed with the InfoNCE estimator. The signal is on the order of millibits.
        (B)~Regional information footprint, where mutual information is estimated for entire binding regions rather than individual sites. The collective signal is about an order of magnitude larger, providing substantially more robustness against noise.}
    \label{fig:footprint}
\end{figure*}

\subsection{Mutual information estimation}\label{sec:estimation}
Estimating mutual information from finite data is a central technical challenge for our framework.
Mutual information $I(A:B)$ measures the statistical dependence between two random variables via the Kullback-Leibler divergence between their joint and product-of-marginals distributions:
\begin{align}
    I(A:B) &= \mathbb{E}_{P(a,b)}\left[\log \frac{P(a,b)}{P(a)P(b)}\right]  \\
    &=: D_{\rm KL}[P(a,b) \| P(a)P(b)]  \nonumber ,
\end{align}
where base-2 logarithms are implied throughout to measure information in bits.
In our setting, $a$ corresponds to the compressed word $\bm{T}$ (or, at the single-site level, a base $B_i$) and $b$ to the expression level $\mu$.

\subsubsection{Limitations of histogram-based estimation.}
A common approach is to bin the continuous expression level $\mu$ into discrete states (\emph{e.g.}, $\mu \in \{0,1\}$ representing off and on) and estimate entropies from normalised histograms $Q(\mu_{\rm bin}) = \text{count}(\mu_{\rm bin})/m$, yielding
\begin{align}
    H(\mu) &= -\sum_\mu P(\mu) \log P(\mu) \approx -\sum_{\mu_{\rm bin}} Q(\mu_{\rm bin}) \log Q(\mu_{\rm bin}).
\end{align}
While simple, this approach has two drawbacks.
First, binning voluntarily erases information---dividing a continuous distribution into a few discrete bins loses details that may be relevant for identifying binding sites.
Second, the per-site mutual information signal $I(B_i:\mu)$ is typically on the order of millibits [Fig.~\ref{fig:footprint}(A)], making reliable estimation difficult with realistic sample sizes even when the improved InfoNCE estimator~\cite{infonce} is used in place of naive histograms.

\subsubsection{Variational lower bounds.}
We instead exploit a variational representation of mutual information that avoids binning entirely.
The key idea is that if $I(A:B)$ is large, samples from the joint distribution $P(a,b)$ should be distinguishable from independently shuffled pairs drawn from $P(a)P(b)$.
Given iid samples $[(a_i, b_i)]_{i=1}^m$ from $P(a,b)$, any cross-pairing $(a_i, b_{j\neq i})$ is distributed according to $P(a)P(b)$. High mutual information means a classifier $f$ can reliably separate the two:
\begin{equation}
    f(a_i, b_i) > 0,\quad f(a_i, b_{j\neq i}) < 0.
\end{equation}
This intuition is formalised by the Donsker-Varadhan representation~\cite{donsker_varadhan_1983, mine, poole2019variationalboundsmutualinformation}:
\begin{equation}\label{eq:DV}
    I(A:B) = \sup_{f\in\mathcal{F}} \left\{ \mathbb{E}_{P(a,b)}[f(a,b)] - \log \mathbb{E}_{P(a)P(b)}[e^{f(a,b)}] \right\},
\end{equation}
where the supremum is taken over a function class $\mathcal{F}$. For any fixed $f$, the expression inside the braces is a lower bound on the mutual information (see Appendix for a self-contained proof). By parameterizing $f$ as a neural network (a ``critic'') and optimizing via stochastic gradient descent, one obtains a tight, differentiable lower bound on $I(A:B)$.

In practice, we use the InfoNCE estimator~\cite{infonce}, where the expectations in Eq.~\eqref{eq:DV} are computed as Monte Carlo averages over minibatches. This estimator has two properties essential for our framework.
First, it operates directly on continuous expression values, eliminating the need for binning and the associated information loss.
Second, the bound is differentiable with respect to the compression map $\Lambda$, so mutual information can serve simultaneously as the objective function for optimizing the compressed words $\bm{T}$.

\subsubsection{From local to global information.}
The information footprint approach~\cite{Kinney2010Final, Ireland2020} computes the per-site mutual information $I(B_i : \mu)$ independently for each position, flagging sites above a threshold $\epsilon$:
\begin{equation}\label{eq:infofootprint}
    \left\{i:\, I(B_i: \mu) > \epsilon \right\}.
\end{equation}
As shown in Fig.~\ref{fig:footprint}(A), this signal is on the order of millibits for a synthetic simple-repression architecture.
Using the variational estimator, we can go beyond single sites and directly estimate the mutual information of entire binding regions, $I(\{B_i \in \text{binding region}\}: \mu)$, without histogram construction.
The collective signal of entire binding regions is an order of magnitude larger than that of individual sites, providing substantially more robustness against noise [Fig.~\ref{fig:footprint}(B)].

The informational blueprint takes this to its logical conclusion: rather than computing mutual information for pre-selected regions, we optimize the compression map $\Lambda$ itself to discover which regions matter. This amounts to the transition from a local to a global information-theoretic probe~\cite{PhysRevE.104.064106}:
\begin{equation}\label{eq:proposal_informal}
    \underbrace{\arg\max_i \left\{I(B_i: \mu) \right\}}_{\text{local footprint}}\quad \rightarrow\quad \underbrace{\arg\max_{\Lambda}\left\{I(\Lambda \circ \bm{B}: \mu) \right\}}_{\text{global compression}},
\end{equation}
where $\Lambda: \mathbb{Z}_{2}^N \rightarrow \mathbb{Z}_2^n$ compresses the full binary mutation vector into a short word $\bm{T}$ with $n\ll N$.

\subsubsection{Critic architecture and training details.}
The critic function $f$ in the InfoNCE estimator is parameterized as a 2-layer MLP with 64 hidden units and ReLU activations. The threshold function $\sigma$ in the compression is approximated during training using the straight-through estimator. We optimize using Adam with learning rate $10^{-3}$ for $10^4$ steps with mini-batch size 256. Each optimization is repeated from 20 random initializations; we report filters from the run achieving highest $I(\bm{T}:\mu)$.

\subsection{Solving the variational compression problem with different trial functions}\label{sec:trial_functions}

The compression in Eq.~\eqref{eq:envelope_filter} defines a variational problem: the mutual information $I(\bm{T}; \mu)$ is maximized over the space of trial filters $\Lambda_{\nu i}$. As in all variational principles, the design of the trial function directly influences the practical performance of the optimization. An overly restrictive ansatz may fail to capture relevant structure, while a fully unconstrained parameterization can be difficult to optimize in the presence of noise. In this section, we describe three progressively more sophisticated parameterizations of the filter---from fully unconstrained weights to scalar-amplitude Gaussians to the envelope factorization introduced in Section~\ref{sec:bio_prior}---each offering a different trade-off between expressiveness and regularization.

\subsubsection{Unconstrained linear filters}\label{sec:localised_filters}
In the simplest parameterization, each filter component $\Lambda_{\nu i}$ is a vector of $N$ freely optimized weights---one per sequence position. This maximally expressive ansatz can, in principle, capture any linear compression of the binary mutation vector $\bm{B}$. However, the $\mathcal{O}(nN)$ free parameters make the optimization susceptible to overfitting, particularly when the mutual information signal is weak relative to experimental noise. In practice, the unconstrained filters nevertheless localize on binding sites (Figs.~\ref{fig:wordlength},~\ref{fig:synthetic}), but they can exhibit noisy tails outside the true binding regions and may fail to resolve additional sites when the information gain per hyperletter becomes small.

\begin{figure}[ht!]
    \centering
    \includegraphics[width=\columnwidth]{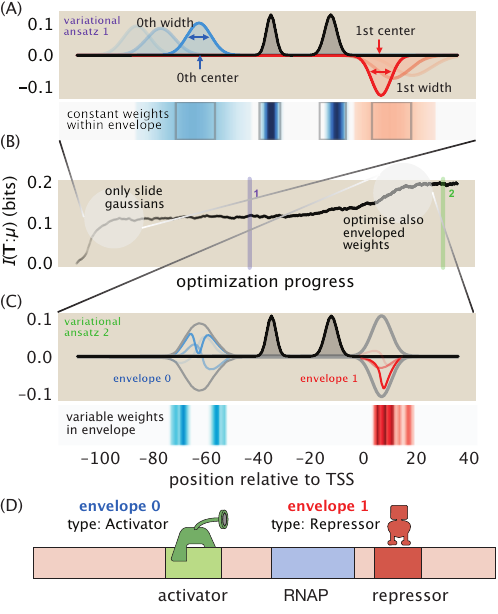}
    \caption{\textbf{Biologically informed variational ansatze for compression filters.}
    (A)~As the simplest way of incorporating biological priors, each envelope component is parameterized by a center $\mu_{\nu,\alpha}$ and width $\sigma_{\nu,\alpha}$; transparency indicates optimization progress from early (light) to late (dark) iterations. The weights remain constant within each envelope.
    (B)~Mutual information $I(T:\mu)$ between the compressed variable and expression level during optimization for two variational ansatze. 
    (C)~A richer ansatz function frees the per-position weights $\lambda_{\nu i}$  within the envelopes and allows the filter to capture detailed structure within a binding site while remaining localized by the envelope. This recovers more information about gene expression than (A).
    (D)~The learned envelope parameters can be read out as binding site positions and widths, 
    here identifying an activator upstream and a repressor downstream of the transcription start site, 
    while the enveloped weights recover the internal structure of each site.
    }
    \label{fig:sliding_filters}
\end{figure}

\subsubsection{Scalar-amplitude Gaussian filters}\label{sec:gaussian_filters}
To incorporate the biological prior that TF binding sites span approximately 15--25\,bp, we constrain each filter to a Gaussian envelope with a scalar amplitude, as shown in Fig.~\ref{fig:sliding_filters}~(A):
\begin{equation}\label{eq:scalar_gaussian}
    \Lambda_{\nu i} = \alpha_\nu \exp\!\left[-\frac{(i - c_\nu)^2}{2w_\nu^2}\right]\,,
\end{equation}
with learnable center $c_\nu$, width $w_\nu$, and amplitude $\alpha_\nu$. This reduces the number of free parameters from $\mathcal{O}(N)$ per filter to just three.
To accommodate bipartite binding sites such as the RNAP $-10$ and $-35$ boxes, each filter can be extended to a mixture of $K$ Gaussian components,
\begin{equation}
    \Lambda_{\nu i} = \sum_{k=1}^{K} m_{\nu,k}\, \alpha_{\nu,k} \exp\!\left[-\frac{(i - c_{\nu,k})^2}{2w_{\nu,k}^2}\right]\,,
\end{equation}
where $m_{\nu,k}$ are mixing weights normalized via softmax (typically $K=2$).
The learnable parameters are constrained as follows: centers $c_{\nu,k}$ are parameterized via sigmoid to lie within $[0, N]$; widths $w_{\nu,k}$ are parameterized via $\tanh$ to lie within $[w_{\min}, w_{\max}]$ (default 10--30\,bp); and amplitudes $\alpha_{\nu,k}$ are unconstrained real values whose sign encodes the regulatory function (positive for activators, negative for repressors relative to RNAP).

Binding site predictions are read off directly from the learned parameters: $c_{\nu,k}$ gives the site location, $w_{\nu,k}$ gives the footprint width, and $\mathrm{sign}(\alpha_{\nu,k})$ indicates activator versus repressor function.

\subsubsection{Envelope-parameterized filters}\label{sec:enveloped_filters}
The envelope parameterization introduced in Section~\ref{sec:bio_prior} factorizes each filter as
\begin{equation}
    \Lambda_{\nu i} = \mathcal{W}_{\nu i}\,\lambda_{\nu i}\,,
\end{equation}
where $\mathcal{W}_{\nu i} \in [0,1]$ is a smooth localizing envelope and $\lambda_{\nu i}$ are freely optimized per-position weights, as visualised in Fig.~\ref{fig:sliding_filters}~(C). This interpolates between two limits: when $\lambda_{\nu i}$ is constant across positions, it reduces to the scalar-amplitude Gaussian of Eq.~\eqref{eq:scalar_gaussian}; when the envelope width $w_\nu \to \infty$, it recovers the fully unconstrained filter. In practice, the biological prior on width regularizes the optimization landscape sufficiently that filter centers converge reliably to binding site positions, while the free weights resolve internal structure that a scalar amplitude cannot capture, and thereby achieve a larger amount if mutual information, as illustrated in Fig.~\ref{fig:sliding_filters}~(B).

We use two envelope shapes. The Gaussian envelope,
\begin{equation}
    \mathcal{W}_{\nu i}^{\rm gaussian} = \exp\!\left[-\frac{(i - c_\nu)^2}{2w_\nu^2}\right]\,,
\end{equation}
provides smooth localization with learnable center $c_\nu$ and width $w_\nu$. Alternatively, a soft rectangular window is constructed from a product of sigmoids:
\begin{equation}
    \mathcal{W}_{\nu i}^{\rm rectangle}(\tau) = \sigma_\tau\!\left(i - c_\nu + \tfrac{w_\nu}{2}\right)\,\sigma_\tau\!\left(c_\nu + \tfrac{w_\nu}{2} - i\right)\,,
\end{equation}
where $\sigma_\tau(x) = (1 + e^{-x/\tau})^{-1}$ and $\tau$ controls the edge sharpness. The Gaussian provides smoother localization; the sigmoid product produces a flatter passband that weights positions within the window more uniformly. In both cases, the width $w_\nu$ is biased to $\sim 20$--$30$\,bp but remains learnable.
The resulting optimal envelopes for each hyperletter can then be used to directly filter out the underlying sequence of the predicted binding sites at base-pair resolution, see Fig.~\ref{fig:sliding_filters}~(D).

\paragraph{Discrete center updates.}
Since promoter sequences are discrete, we optimize the envelope centers $c_\nu$ using a sign-accumulation rule rather than standard gradient descent. At each step, the sign of $\partial I / \partial c_\nu$ is computed and accumulated over a window of training steps. The center moves by one base pair in the majority direction only when the accumulated vote exceeds a threshold, ensuring that centers lock onto optimal positions rather than jittering due to gradient noise.

\section*{Acknowledgements}
D.E.G acknowledges support by the NSF-Simons National Institute for Theory and Mathematics in Biology (NITMB) Fellowship supported via grants from the NSF (DMS-2235451) and Simons Foundation (MPS-NITMB-00005320). 
R.W.P is grateful to Charles Trimble for support.
T.R. was supported by a fellowship from Boehringer Ingelheim Fonds.
H.G.G. was supported by the National Institutes of Health R35GM158200, a Winkler Scholar Faculty Award, the Chan Zuckerberg Initiative (grant CZIF2024-010479), the Miller Institute for Basic Research in Science, University of California Berkeley. H.G.G is also a Chan Zuckerberg Biohub investigator.
We are grateful to the NIH for support through award numbers DP1OD000217 (Director's Pioneer Award) and NIH MIRA 1R35 GM118043-01.
This research was partly supported from the National Science Foundation through the Center for Living Systems (grant no. 2317138), the National Institute for Theory and Mathematics in Biology, the Simons Foundation and the Chan Zuckerberg Foundation. This work was completed in part with resources provided by the University of Chicago’s Research Computing Center.

\bibliographystyle{unsrt}
\bibliography{bibliography}

\clearpage
\onecolumngrid
\appendix

\newpage
\section{Calculations for the model in Box 1}\label{Appendix:BoxCalculations}

Box~\ref{box:algorithm} provides a toy model of the variational calculations used to find the ``best'' coarse grained description of a highly simplified promoter in which the RNAP site corresponds to a single base pair.  The aim of the toy calculation presented there and the slightly more sophisticated and complete one presented here is to introduce the key concepts and notation used in the full approach carried out in the paper.   As seen in the Box, this calculation is ``variational'' in the sense that we have a number of candidate trial functions. We use each such trial function to compute a score in the form of the mutual information between our hyperletters and the level of gene expression, and out of that collection of trial functions, we pick the winner.  For the cases considered here, there is a very small set of competitors and we explicitly compute the information for each of those competitors and, as shown in the box, compare them and choose the winner! 

This is similar in spirit to many other examples of simplified ``shape functions'' throughout mathematics, physics and engineering.  For example, in the method of linear combination of atomic orbitals (LCAO) in quantum mechanics, one writes a molecular wave function in a highly simplified way as a linear combination of atomic wave functions centered on each of the atoms making up the molecule.  In the finite element method, one discretizes some partial differential equation into a discrete set of nodes and carries out interpolation between those nodes.  In both of these examples, once the simplfied set of candidates has been selected, one then chooses a winner according to some scoring function.  That is the spirit of the toy calculations presented in the Box and explored more deeply here.

\subsection{One-site filter}

As a reminder, Box~\ref{box:algorithm} features a toy model of a promoter sequence containing five sites. Only the fourth   site is involved in RNAP binding. As a result, sequences with mutations at that site will results in a change in gene expression.
Our objective is to show how the information blueprint approach makes it possible to identify this position in a principled way. Specifically, we seek to calculate the mutual information between gene expression and the different hyperletters $T(x)$, each defined by its corresponding filter $\mathbf{\Lambda(x)}$. More concretely, for the $m$th mutant we have
\begin{equation}\label{eq:toy_hyperletters}
    T^{(m)}(x) = \sigma\left(\sum_{i=1}^5  B^{(m)}_i \Lambda_i(x) \right),
\end{equation}
where $B^{(m)}_i=0$ indicates that the base pair at position $i$ of the $m$th mutant is wild type and $B^{(m)}_i=1$ marks the presence of a mutation. Recall that our toy MPRA library consists of five mutants, each with a single mutation at a unique site, that is:
\begin{equation}
\begin{aligned}
& \mathbf{B}^{(m=1)}=(1,0,0,0,0) \qquad \text{first mutant sequence}, \\
& \mathbf{B}^{(m=2)}=(0,1,0,0,0) \qquad \text{second mutant sequence}, \\
& \mathbf{B}^{(m=3)}=(0,0,1,0,0) \qquad \text{third mutant sequence}, \\
& \mathbf{B}^{(m=4)}=(0,0,0,1,0) \qquad \text{fourth mutant sequence}, \\
& \mathbf{B}^{(m=5)}=(0,0,0,0,1) \qquad \text{fifth mutant sequence}.
\end{aligned}
\end{equation}
or more compactly
\begin{equation}\label{eq:toy_sequences}
    B^{(m)}_i = \delta_{im}
\end{equation}

Because of the nature of the particular filters chosen for this example, $x$ corresponds to which site each vector $\mathbf{\Lambda(x)}$ filters out. Namely, the filter is only non-zero at one position given by
\begin{equation}\label{eq:toy_filters}
    \Lambda_i(x) = \delta_{ix},
\end{equation}
where the label $x$ tells us the non-zero site in the filter and the subscript $i$ defines the $i^{th}$ component of the filter vector  $\mathbf{\Lambda(x)}$.
More explicitly, for our 5-base pair promoter,  as shown in 
Figure~\ref{fig:FiltersBox}, the different trial filters are given by the vectors
\begin{equation}
\begin{aligned}
& \mathbf{\Lambda}(x=1)=(1,0,0,0,0) \qquad \text{first trial filter}, \\
& \mathbf{\Lambda}(x=2)=(0,1,0,0,0) \qquad \text{second trial filter}, \\
& \mathbf{\Lambda}(x=3)=(0,0,1,0,0) \qquad \text{third trial filter}, \\
& \mathbf{\Lambda}(x=4)=(0,0,0,1,0) \qquad \text{fourth trial filter}, \\
& \mathbf{\Lambda}(x=5)=(0,0,0,0,1) \qquad \text{fifth trial filter}.
\end{aligned}
\end{equation}
However, it is important to note that, as described in detail in the main text, both the shape of the filters, and which filters are most informative will be subject to optimization in the context of real world data on real promoters and using filters that are way more sophisticated than those used here.

As noted above, we must have a scheme for scoring each of our variational trial functions.  Throughout the paper, that scoring function is obtained by computing the mutual information between a given hyperletter and the level of gene expression, which for our toy model defined as a function of the filter position $x$
\begin{equation}
    I(\mu : T(x)) = \underbrace{H(\mu)}_{\mbox{initial uncertainty}} - \underbrace{H(\mu |T(x))}_{\mbox{remaining uncertainty}},
    \label{eqn:InformationScore}
\end{equation}
where $\mu$ corresponds to the level of gene expression. As a reminder, the entropy $H(\mu$) is given by
\begin{equation}
H(\mu) =-\sum_{\text {expression }} P(\mu) \log_2 (P(\mu)).
\end{equation}
Given a mutant sequence $\mathbf{B}^{(m)}$, Eq.~\ref{eq:toy_hyperletters} (that is, Eq.~\ref{eq:box_1} )
gives us a recipe to
calculate hyperletter for each mutant $m$, $T^{(m)}(x)$. For completeness, we reproduce the equation here by combining Eq.~\ref{eq:toy_sequences} for the $m$th mutant sequence $\mathbf{B}^{(m)}$, and Eq.~\ref{eq:toy_filters} for the trial filter $\Lambda(x)$: 
\begin{equation}
    T^{(m)}(x) = \sigma \left( \mathbf{B}^{(m)} \cdot \mathbf{\Lambda}(x) \right) =\sigma\left(\sum_{i=1}^5 B_i^{(m)} \Lambda_i(x) \right) = \sigma\left(\sum_{i=1}^5 \delta_{im} \delta_{ix}\right) = \sigma(\delta_{mx})=\begin{cases} 1, \quad m=x, \\ 0, \quad m\neq x \end{cases}.
    \end{equation}
The function $\sigma \left( \mathbf{B}^{(m)} \cdot \mathbf{\Lambda} \right)$ performs a thresholding operation by producing a score of its own by evaluating the magnitude of its argument $\mathbf{B}^{(m)} \cdot \mathbf{\Lambda}$.
For example, the value of the first hyperletter, for the first sequence is
\begin{equation}
\begin{aligned}
T^{(m=1)}(x=1) &= \sigma\left(\mathbf{B}^{(m=1)}  \cdot \mathbf{\Lambda}(x=1)\right) =\sigma\left((1,0,0,0,0)  \cdot (1,0,0,0,0)\right) \\
& =\sigma(1)=\sigma(\text {above threshold}) =1.
\end{aligned}
\end{equation}
In contrast, the value of the first hyperletter for the second sequence is
\begin{equation}
\begin{aligned}
T^{(m=2)}(x=1) &= \sigma\left(\mathbf{B}^{(m=2)}  \cdot \mathbf{\Lambda}(x=1)\right) =\sigma\left((0,1,0,0,0)  \cdot (1,0,0,0,0)\right)  \\
& =\sigma(0)=\sigma(\text {below threshold}) =0.
\end{aligned}
\end{equation}

To compute the total information score represented by eqn.~\ref{eqn:InformationScore}, we need to quantify how much knowing the hyperletter value reduces our uncertainty about expression. This is precisely the conditional entropy, defined by
\begin{equation}
    H(\mu |T(x)) \,=\, P(T(x)=0)  H (\mu |T(x)=0) \,+\, P(T(x)=1)  H (\mu |T(x)=1).
\end{equation}
The conditional entropy is a weighted average over the two possible hyperletter outcomes: for each value $T(x) \in \{0,1\}$, we compute the entropy of expression levels among only those sequences sharing that hyperletter value, then weight by the probability of that outcome occurring. If the filter has correctly identified a binding site, then conditioning on whether that site is intact ($T=0$) or disrupted ($T=1$) partitions the sequences into groups with more homogeneous expression levels --- meaning each $H(\mu | T(x) = 0)$ and $H(\mu | T(x) = 1)$ is smaller than our initial uncertainty $H(\mu)$ before any hyperletter was revealed.

We now proceed to calculate all the terms necessary to compute the mutual information. As seen in the Box, out of the five mutated sequences, four of them have high expression ($h$) and one of them has low expression ($l$).  As a result, the probability that gene expression is high  is given by
\begin{equation}
P(\mu)=\left\{\begin{array}{l}
P(\mu=h)=4 / 5 \\
P(\mu=l)=1 / 5
\end{array}\right.
\end{equation}
From these probabilities, we can compute the entropy associated with gene expression as
\begin{equation}
\begin{aligned}
H(\mu) & =-\sum_{\text {expression }} P(\mu) \log_2 (P(\mu)) \\
& =-[4 / 5 \log_2 (4 / 5)+1 / 5 \log_2 (1 / 5)].
\end{aligned}
\end{equation}

Next, we compute the conditional entropy for the different hyperletters. We begin with the first hyperletter, $T(x=1)$, such that
\begin{equation}\label{eq:HmuT1}
    H(\mu |T(x=1)) \,=\, P(T(x=1)=0)  H (\mu |T(x=1)=0) \,+\, P(T(x=1)=1)  H (\mu |T(x=1)=1).
\end{equation}
An examination of the table in the Box tells us that
\begin{equation}
    P(T(x=1)=0) = 4/5
\end{equation}
and 
\begin{equation}
    P(T(x=1)=1) = 1/5.
\end{equation}
Further, the conditional entropies on the right hand side of equation~\ref{eq:HmuT1} are
\begin{equation}
\begin{aligned}
H(\mu |T(x=1)=1)= & -\sum_{\mbox{expression}} P(\mu |T(x=1)=1) \log_2 \left( P(\mu |T(x=1)=1) \right)\\
= & -\left( p\left(\mu=h | T(x=1)=1\right) \log_2 (P(\mu=h | T(x=1)=1)\right. \\
& \left. + P(\mu=l | T(x=1)=1) \log_2 (P(\mu=l | T(x=1)=1) \right) \\
= & - \left( 1 \cdot \log_2 (1)+0 \cdot \log_2 (0)\right)=0,
\end{aligned}
\end{equation}
and
\begin{equation}
\begin{aligned}
H(\mu |T(x=1)=0)= & -\sum_{\mbox{expression}} P(\mu |T(x=1)=0) \log_2 \left( P(\mu |T(x=1)=0) \right)\\
= & -\left( p\left(\mu=h | T(x=1)=0\right) \log_2 (P(\mu=h | T(x=1)=0)\right. \\
& \left. + P(\mu=l | T(x=1)=0) \log_2 (P(\mu=l | T(x=1)=0) \right) \\
& =-\left(\frac{1}{4}\log_2 {1 \over4} +\frac{3}{4} \log_2 \frac{3}{4} \right) \\
& =\frac{1}{4} \log_2 4+\frac{3}{4} \log_2 \frac{4}{3}.
\end{aligned}
\end{equation}
Putting this all together we arrive at
\begin{equation}
\begin{aligned}
    H(\mu |T(x=1)) \,&=\, \underbrace{P(T(x=1)=0)}_{4/5}  \underbrace{H (\mu |T(x=1)=0)}_{\frac{1}{4} \log_2 4+\frac{3}{4} \log_2 \frac{4}{3}} \,+\, \underbrace{P(T(x=1)=1)}_{1/5}  \underbrace{H (\mu |T(x=1)=1)}_{0} \\
    &= {4 \over 5} \left(\frac{1}{4} \log_2 4+\frac{3}{4} \log_2 \frac{4}{3}\right).
\end{aligned}
\end{equation}

Assembling all of the pieces, we can now compute the score associated with the hyperletter $T(x=1)$ which is given by
\begin{equation}
\begin{aligned}
 I(T(x=1): \mu)= & H(\mu)-H(\mu |T(x=1)) \\
& \quad=-\left({4 \over 5} \log_2 {4 \over 5} + {1 \over 5} \log_2 {1 \over 5} \right) - {4 \over 5} \left(\frac{1}{4} \log_2 4+\frac{3}{4} \log_2 \frac{4}{3}\right) \\
 & \quad \approx 0.722 -0.649 \\
& \quad \approx 0.073.
\end{aligned}
\end{equation}
This score is plotted as the first bar in the bar graph shown in Figure~\ref{fig:box_mi} in the box.  
If we repeat this same calculation for the mutual information between the gene expression level $\mu$ and the hyperletters, $T(x=2)$, $T(x=3)$, and $T(x=5)$,
we find the exact same score.  There was nothing distinguishing any of these hyperletters, and thus in the histogram of scores, they all have the same mutual information.

The only case that potentially has a different mutual information score is for $T(x=4)$.
To that end, we compute the conditional entropy for hyperletter, $T(x=4)$, such that
\begin{equation}\label{eq:HmuT4}
    H(\mu |T(x=4)) \,=\, P(T(x=4)=0)  H (\mu |T(x=4)=0) \,+\, P(T(x=4)=1)  H (\mu |T(x=4)=1).
\end{equation}
An examination of the table in the Box tells us that
\begin{equation}
    P(T(x=4)=0) = 4/5
\end{equation}
and 
\begin{equation}
    P(T(x=4)=1) = 1/5.
\end{equation}
Further, the conditional entropies on the right hand side of equation~\ref{eq:HmuT4} are
\begin{equation}
\begin{aligned}
H(\mu |T(x=4)=1)= & -\sum_{\mbox{expression}} P(\mu |T(x=4)=1) \log_2 \left( P(\mu |T(x=4)=1) \right)\\
= & -\left( p\left(\mu=h | T(x=4)=1\right) \log_2 (P(\mu=h | T(x=4)=1)\right. \\
& \left. + P(\mu=l | T(x=4)=1) \log_2 (P(\mu=l | T(x=4)=1) \right) \\
= & - \left( 0 \cdot \log_2 (0) + 1 \cdot \log_2 (1)\right)=0,
\end{aligned}
\end{equation}
and
\begin{equation}
\begin{aligned}
H(\mu |T(x=4)=0)= & -\sum_{\mbox{expression}} P(\mu |T(x=4)=0) \log_2 \left( P(\mu |T(x=4)=0) \right)\\
= & -\left( p\left(\mu=h | T(x=4)=0\right) \log_2 (P(\mu=h | T(x=4)=0)\right. \\
& \left. + P(\mu=l | T(x=4)=0) \log_2 (P(\mu=l | T(x=4)=0) \right) \\
= & - \left( 1 \cdot \log_2 (1) + 0 \cdot \log_2 (0)\right)=0.
\end{aligned}
\end{equation}
Putting this all together we arrive at
\begin{equation}
\begin{aligned}
    H(\mu |T(x=4)) \,&=\, \underbrace{P(T(x=4)=0)}_{4/5}  \underbrace{H (\mu |T(x=4)=0)}_{0} \,+\, \underbrace{P(T(x=1)=1)}_{1/5}  \underbrace{H (\mu |T(x=1)=1)}_{0} \\
    &= 0.
\end{aligned}
\end{equation}
Assembling all of the pieces, we can now compute the score associated with the hyperletter $T(x=4)$ which is given by
\begin{equation}
\begin{aligned}
& I(\mu : T(x=4))=H(\mu)-H(\mu |T(x=4)) \\
& \quad \approx 0.722.
\end{aligned}
\end{equation}
As shown in the bar graph in Figure~\ref{fig:box_mi}, this hyperletter contains the maximum mutual information.  Hence, $T(x=4)$ is the maximally informative hyperletter.  This raises the
question of how our score depends upon number of hyperletters in our information blueprint.  To see how to address that question, we turn to filters for our toy promoter centered on more than one basepair.

\subsection{Two-site filter}

In our second toy model, we enlarge the space of filters to include more than one site. Specifically, our goal is to compute the mutual information between expression level and the word ${\bf T}$ given by the combination of the hyperletters $T(x=1)$ and $T(x=4)$. The mutual information between this word and gene expression is given by
\begin{equation}
I(\mu : \boldsymbol{T})=\underbrace{H(\mu)}_{\text {initial uncertainty }}-\underbrace{H(\mu |\boldsymbol{T})}_{\text {remaining uncertainty }}.
\end{equation}
We have already calculated the initial uncertainty $H(\mu)$. What is left for us to calculate is the conditional entropy representing the remaining uncertainty, $H(\mu |\boldsymbol{T})$.

We calculate this conditional entropy by exploring all the possible values the $T(x=1)$ and $T(x=4)$ hyperletters can adopt. More precisely, we must evaluate the sum
\begin{equation}\label{eq:CondProb2Site}
\begin{aligned}
 H(\mu |\boldsymbol{T}) = &H(\mu |(T(x=1), T(x=4))= \\
 = &P((T(x=1)=0, T(x=4)=0)  \,  H(\mu |(T(x=1)=0, T(x=4)=0))+ \\
& P((T(x=1)=1, T(x=4)=0) \, H(\mu |(T(x=1)=1, T(x=4)=0))+ \\
& P((T(x=1)=0, T(x=4)=1) \,  H(\mu |(T(x=1)=0, T(x=4)=1))+ \\
& P((T(x=1)=1, T(x=4)=1) \, H(\mu | (T(x=1)=1, T(x=4)=1)).
\end{aligned}
\end{equation}
The table in the Box in the main text tells us the values of the different probabilities in equation~\ref{eq:CondProb2Site}. Specifically, we have
\begin{equation}
P((T(x=1)=0, T(x=4)=0) = 3/5,
\end{equation}
\begin{equation}
P((T(x=1)=1, T(x=4)=0) = 1/5,
\end{equation}
\begin{equation}
P((T(x=1)=0, T(x=4)=1) = 1/5,
\end{equation}
and
\begin{equation}
P((T(x=1)=1, T(x=4)=1) = 0,
\end{equation}
each of which can be read off by counting how many of the table entries correspond to each of those conditions.
As a result, equation~\ref{eq:CondProb2Site} can be rewritten as
\begin{equation}
\begin{aligned}
 H(\mu |\boldsymbol{T}) = &H(\mu |(T(x=1), T(x=4)) \\
 = &3/5  \,  H(\mu |(T(x=1)=0, T(x=4)=0))+ \\
& 1/5 \, H(\mu |(T(x=1)=1, T(x=4)=0))+ \\
& 1/5 \,  H(\mu |(T(x=1)=0, T(x=4)=1)).
\end{aligned}
\end{equation}

To make progress, we now calculate each conditional entropy. We begin with
\begin{equation}
\begin{aligned}
& H(\mu | T(x=1)=0, T(x=4)=0)) \\
& =P(\mu=h | T(x=1)=0, T(x=4)=0)) \, \log_2 \left(P(\mu=h | T(x=1)=0, T(x=4)=0)) \right) \\
& +P(\mu=l | T(x=1)=0, T(x=4)=0)) \, \log_2 \left( P(\mu=l | T(x=1)=0, T(x=4)=0)) \right) \\
& = 1 \cdot \log_2 1 + 0 \cdot \log_2 0\\
& =0.
\end{aligned}
\end{equation}
The next term is
\begin{equation}
\begin{aligned}
& H(\mu |T(x=1)=1, T(x=4)=0)) \\
& =P(\mu=h | T(x=1)=1, T(x=4)=0)) \, \log_2 \left(P(\mu=h | T(x=1)=1, T(x=4)=0)) \right) \\
& +P(\mu=l | T(x=1)=1, T(x=4)=0)) \, \log_2 \left( P(\mu=l | T(x=1)=1, T(x=4)=0)) \right) \\
& = 1 \cdot \log_2 1 + 0 \cdot \log_2 0\\
& =0.
\end{aligned}
\end{equation}
The third and final term is
\begin{equation}
\begin{aligned}
& H(\mu |T(x=1)=0, T(x=4)=1)) \\
& =P(\mu=h | T(x=1)=0, T(x=4)=1)) \, \log_2 \left(P(\mu=h | T(x=1)=0, T(x=4)=1)) \right) \\
& +P(\mu=l | T(x=1)=0, T(x=4)=1)) \, \log_2 \left( P(\mu=l | T(x=1)=0, T(x=4)=1)) \right) \\
& = 0 \cdot \log_2 0 + 1 \cdot \log_2 1\\
& =0.
\end{aligned}
\end{equation}

Putting this all together, we arrive at
\begin{equation}
 H(\mu |\boldsymbol{T}) = H(\mu |(T(x=1), T(x=4)) = 0.
\end{equation}
Such that the mutual information between the word $\mathbf{T}$ containing the $T(x=1)$ and $T(x=4)$ hyperletters is
\begin{equation}
I(\mu : \boldsymbol{T})=H(\mu)-H(\mu |\boldsymbol{T}) = H(\mu) \approx 0.722 = I(\mu : T(x=4)).
\end{equation}
Our calculation reveals that the mutual information between gene expression and the hyperword containing the $T(x=1)$ and $T(x=4)$ hyperletters is the same as the mutual information between gene expression and the $T(x=4)$ hyperletter alone. In other words, we find that the hyperletter $T(x=4)$ already tells us everything there is to learn about the connection betweeen sequence and expression
for this toy promoter. As a result, adding the $T(x=1)$ hyperletter does not increase our information in any way. 
This is reflected in the plot in Figure~\ref{fig:box_mi}. We see that having a word with a length of one hyperletter---in this case that hyperletter is $T(x=4)$---already saturates the largest possible value of mutual information between gene expression and \emph{any} word.
Since the conditional entropy $H(\mu | T(x=4))$ is already zero, the identity of the additional hyperletters at each subsequent word length is immaterial: no combination can further reduce an uncertainty that has already vanished. Making the word longer by adding any of the other available hyperletters does not improve our overall informational standing.

\newpage
\section{More details on the results for the synthetic datasets}\label{sec:detail_synthetic}

To benchmark the informational blueprint method in a controlled setting where the ground truth is known exactly, we construct synthetic MPRA datasets using sequence-specific thermodynamic models of gene regulation. The approach follows Pan et al.~\cite{Pan2024}, who developed a computational pipeline that maps promoter sequences to expression levels via states-and-weights diagrams derived from experimentally measured energy matrices. Starting from a wild-type promoter, a library of mutant sequences is generated by randomly introducing point mutations at a specified rate. For each mutant, the probability of RNAP being bound---and hence the expression level---is computed from the thermodynamic model, yielding a synthetic dataset in the same format as an RNA-Seq--based MPRA. This procedure can be applied to any regulatory architecture (constitutive, simple repression, activation, double repression, etc.) by specifying the relevant transcription factors, their copy numbers, binding energies, and the logic of their interactions. Crucially, because every parameter is known by construction, the synthetic datasets provide an unambiguous test of whether the blueprint method correctly identifies binding site locations, signs, and counts. Below we present the blueprint results for each architecture in turn.

\subsection{Constitutive promoter}\label{app:constitutive}
First we consider a constitutive promoter regulated only by RNAP.
The optimal compression filter automatically identifies the RNAP binding site: it has strong coupling exclusively at the known binding domains with the correct bipartite structure~\cite{ptashne_genes_2002}, and at least an order of magnitude smaller couplings everywhere else (Fig.~\ref{fig:filter_constitutive_promoter}). Adding a second filter does not increase the retained mutual information, correctly indicating that there is only one functional regulatory element. Notably, the two spatially separated DNA-binding domains of RNAP (the $-35$ and $-10$ boxes) are not split across filters but captured together, reflecting their cooperative function as a single transcription factor.

\begin{figure}[ht!]
    \centering
    \includegraphics[width=0.6\columnwidth]{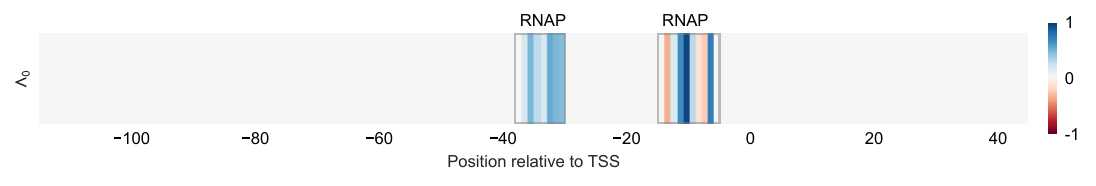}
    \caption{\textbf{The linear optimal compression filter for a constitutive promoter architecture.} The entries of the filter  $\Lambda^0$ are colour-coded according to their magnitude. The filter has strong coupling (and mostly with the same sign) at the RNAP binding sites, and negligible couplings everywhere else.}
    \label{fig:filter_constitutive_promoter}
\end{figure}

The localisation of the filter on the RNAP binding site (Fig.~\ref{fig:filter_constitutive_promoter}) follows from the general decomposition in Section~\ref{sec:covariance}. For a constitutive promoter, expression is a monotonic function of the RNAP binding energy alone, $\mu=\mu(E_P)$, where $E_P(s)=\sum_{i\in P}\varepsilon_{i,s_i}^{(P)}$ sums the position-specific energy contributions within the binding site~$P$. A mutation at position~$i\in P$ shifts the binding energy by an amount that depends on the identity of the mutant base; averaging over the three equally likely substitutions defines
\begin{equation}\label{eq:deltaE}
  \delta\bar{E}_i^{(P)} \;=\; \frac{1}{3}\sum_{s'\neq s_i^{\mathrm{wt}}} \bigl(\varepsilon_{i,s'}^{(P)}-\varepsilon_{i,s_i^{\mathrm{wt}}}^{(P)}\bigr).
\end{equation}
For positions outside the binding site, $B_i$ is uncorrelated with $E_P$ and hence with~$\mu$. Applying Eq.~\eqref{eq:cov_decomp} to this single-TF case gives
\begin{equation}\label{eq:cov_constitutive}
  C_i \;=\; \chi_P\;\delta\bar{E}_i^{(P)}\;\mathbf{1}_{i\in P},
\end{equation}
where $\chi_P=\langle\partial\mu/\partial E_P\rangle_{\mathrm{eff}}$ is the thermodynamic susceptibility. Since $\Lambda_i\propto C_i$ (Section~\ref{sec:covariance}), the optimal filter is $\Lambda_i\propto\delta\bar{E}_i^{(P)}\,\mathbf{1}_{i\in P}$: it is non-zero only within the binding site, and its relative magnitudes reflect how strongly a mutation at each position disrupts binding. The filter thus acts as a direct readout of the DNA--protein contact.

\subsection{Simple activation}\label{app:activation} 
We consider a promoter with a single activator, using parameters that emulate CRP activation at the \emph{lacZYA} promoter. Unlike the constitutive promoter, the information saturates at $n=2$. The learned filters (Fig.~\ref{fig:filter_activation}) couple to the CRP and RNAP binding sites, with the CRP site showing the expected bipartite structure of the homodimer~\cite{ptashne_genes_2002}. Because CRP and RNAP act cooperatively---CRP binding enhances RNAP recruitment---mutations in either site shift expression in the same direction, and the filter accordingly assigns weights of the same sign to both regions. This contrasts with simple repression (Fig.~\ref{fig:filter_repression}), where the RNAP and repressor weights have opposite signs.

\begin{figure}[ht!]
    \centering
    \includegraphics[width=0.6\columnwidth]{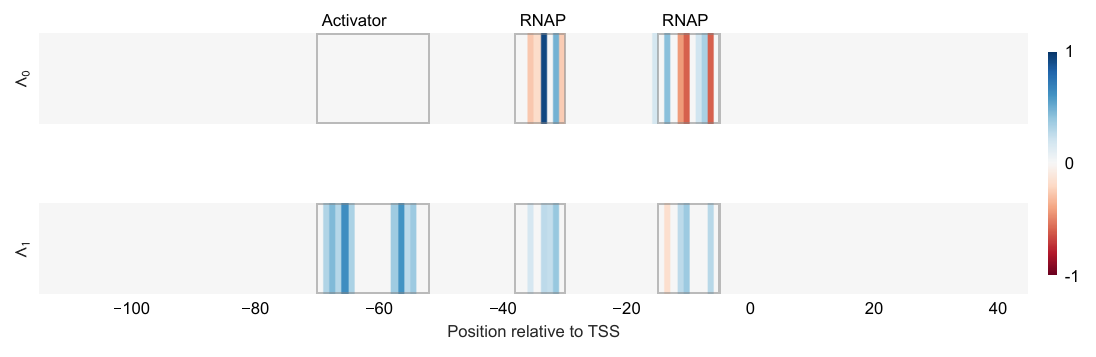}
    \caption{\textbf{Optimal compression for CRP activation.}
        The optimal linear compression map $\Lambda=[\Lambda^0, \Lambda^1]$ filters out the correct binding regions for a simple activation architecture. While the filter $\Lambda^0$ captures the binding sites of RNAP, $\Lambda^1$ couples to the activator and the RNAP binding sites. Both filters have vanishing couplings everywhere else.}
    \label{fig:filter_activation}
\end{figure}

\subsection{Simple repression}\label{app:repression}
Now we look at the case of the simple repression architecture. The synthetic dataset uses the known copy number and binding energy of the lacl repressor.

\begin{figure}[ht!]
    \centering
    \includegraphics[width=0.6\columnwidth]{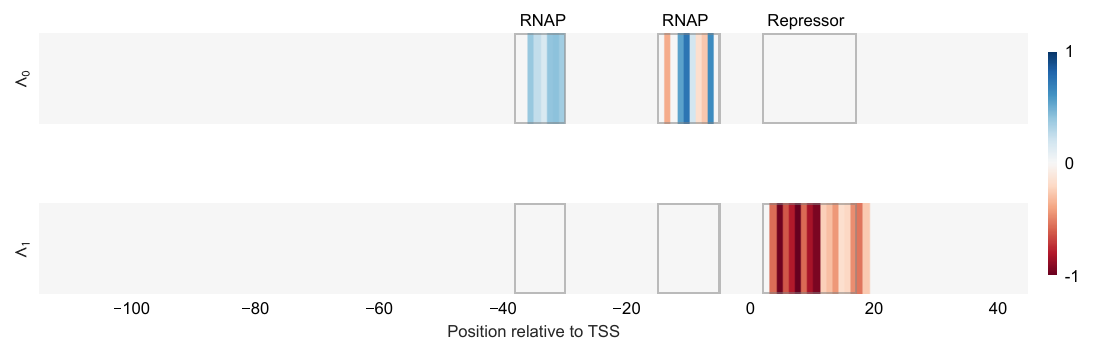}
    \caption{\textbf{Optimal compression for simple repression.}
        The optimal linear compression map $\Lambda=[\Lambda^0, \Lambda^1]$ filters out the correct binding regions for a simple repression architecture. Filters $\Lambda^0$ and $\Lambda^1$ respectively capture the binding sites of RNAP and the repressor. Both filters have vanishing couplings everywhere else.}
    \label{fig:filter_repression}
\end{figure}

In the case of single filter, there is strong coupling to both RNAP with the correct bipartite structure~\cite{ptashne_genes_2002} and repressor binding sites, but with an opposite sign. This indicates that mutations in these two regions lead to opposite expression shifts.
Adding the second filter increases retained information. In other words, compressing the mutations on repression and RNAP binding sites into two separate binary symbols is more informative about the expression shifts than to combine them into a single binary variable. The compressed information saturates at two filters, again matching the number of distinct TFs in the system.

\subsubsection{Effect of expression noise and library size}
\begin{figure*}[hbt!]
    \centering
    \includegraphics[width=0.9\textwidth]{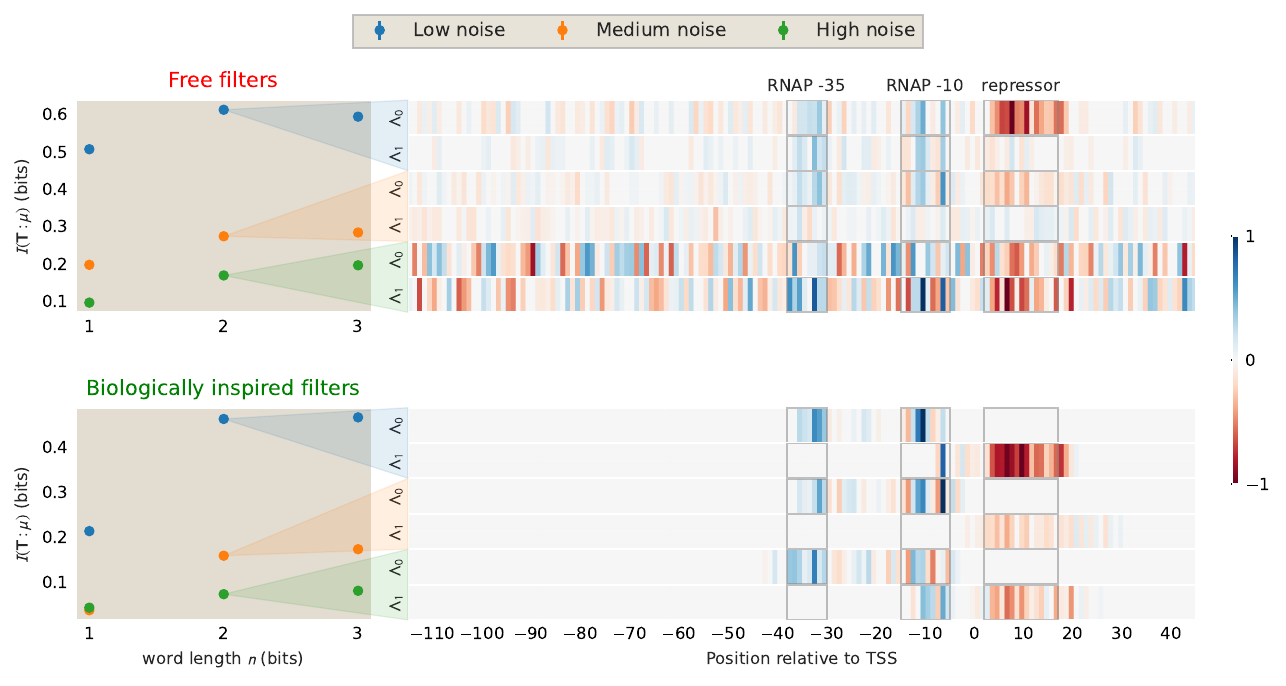}
    \caption{%
        \textbf{Robustness of binding site detection to expression noise.}
        Synthetic simple-repression libraries were generated with low, medium, and high levels of lognormal noise added to the RNA/DNA ratio.
        \textit{Left panels:} captured mutual information $I(\mathbf{T};\mu)$ as a function of word length~$n$ for free (top) and biologically inspired (bottom) filter parameterisations;
        shaded bands indicate the null floor estimated from permuted expression labels.
        \textit{Right panels:} learned two-component filters ($n=2$) at each noise level, sign-calibrated so that activating elements (RNAP binding sites) appear positive (red) and repressing elements appear negative (blue);
        grey boxes mark the true binding site locations.
        The biologically inspired parameterisation maintains sharper, better-localised filters and higher mutual information at elevated noise, extending the regime in which individual binding sites can be resolved.
}
    \label{fig:noise_benchmark}
\end{figure*}

In real experiments the RNA/DNA ratio measurements are noisy, and the signal of individual binding sites can be on the order of millibits, making it difficult to resolve binding sites with unconstrained filters.
To assess robustness of our method, we generate noisy synthetic libraries that mimic the measurement variability of empirical MPRAs. Each promoter variant is expanded into a random number of barcode replicates (drawn from a negative binomial distribution), and log-normal noise of increasing magnitude $\sigma$ is applied to the RNA/DNA expression ratio to represent cell-to-cell and replicate-to-replicate variability.

Figure~\ref{fig:noise_benchmark} summarises the results. As noise grows, the captured mutual information decreases and the filters become increasingly diffuse. In the free parameterisation (top row), the signal-to-noise ratio degrades rapidly: at high noise the information plateau is no longer visible and the filters fail to localise on the binding sites. The biologically informed envelope parameterisation (Section~\ref{sec:bio_prior}) extends the resolvable regime considerably (bottom row): by constraining each filter to a physically plausible width and position, the number of free parameters is drastically reduced, allowing the optimisation to distinguish genuine binding-site signals from noise even when unconstrained filters fail. This highlights the importance of incorporating domain-specific priors to regularize the optimization landscape and enhance robustness against experimental noise, a critical consideration for real-world applications of the blueprint method.

\begin{figure}[hbt!]
    \centering
    \includegraphics[width=0.4\linewidth]{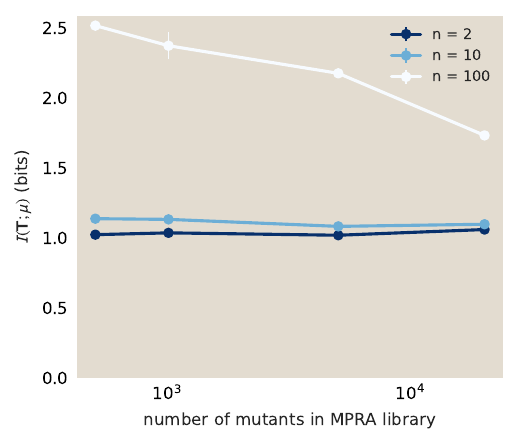}
    \caption{The capture mutual information $I(\bm{T};\mu)$ as a function of the number of mutants $N$ in a simulated simple-repression library, at three word lengths. 
    At $n=2$ (the correct plateau) and $n=10$ (just past saturation), $I(\bm{T};\mu)$ is essentially independent of $N$ down to $N \sim 500$. At $n=100$ (heavily over-parameterised), the InfoNCE bound is positively biased at small $N$ and decreases monotonically as the library grows.}
    \label{fig:libsize}
\end{figure}

A complementary practical concern is the size of the MPRA library itself. To probe this, we drew nested random subsets of size $N$ from a large synthetic simple-repression dataset and re-ran the word-length scan at each $N$. 
As shown in Figure~\ref{fig:libsize}, while the correct plateau ($n=2$), $I(\bm{T};\mu)$ is essentially present even for very small MPRA libraries down to a few hundred mutants, finite-sample artefacts become significant in the largely over-parameterised regime when the number of hyperletters become on the order of the length of the promoter sequence itself ($n=100 \gg n^{\ast}$).
The InfoNCE bound is positively biased at small $N$ and decays monotonically as $N$ grows. 
Library size therefore matters most when $n$ is pushed far above the true number of binding sites.

\subsubsection{Overlapping binding sites}
\begin{figure}[hbt!]
    \centering
    \includegraphics[width=0.5\columnwidth]{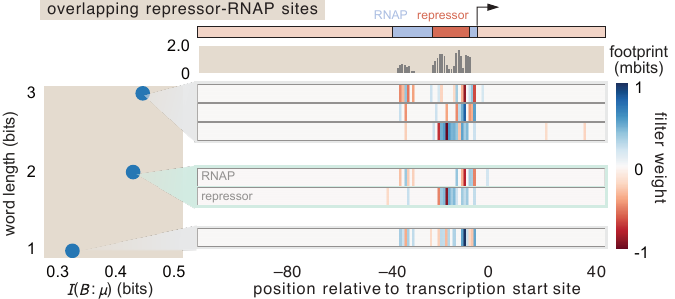}
    \caption{\textbf{Delineating overlapping binding sites.}
        A synthetic promoter with RNAP and repressor binding sites that overlap. At $n=1$, the compression bottleneck is too tight, forcing RNAP (blue) and Repressor (red) into a single conflated filter. Increasing to $n=2$ allows resolving them into two distinct components ($\Lambda_1, \Lambda_2$).
        }
    \label{fig:overlap}
\end{figure}

A common variant of simple repression is when the repressor binding site overlaps with RNAP, since a primary mechanism of repression is steric occlusion. This poses a significant challenge for information footprints~\cite{Pan2024}: when the overlap exceeds roughly half of the binding site length, the opposing signals from RNAP and repressor cancel, making it difficult to identify either site. The blueprint overcomes this because the two-component compression ($n=2$) can assign RNAP and repressor to distinct filters even when their binding sites share common positions. Figure~\ref{fig:overlap} illustrates this for a synthetic promoter with overlapping RNAP and repressor sites.

\begin{figure}[hbt!]
    \centering
    \includegraphics[width=0.85\textwidth]{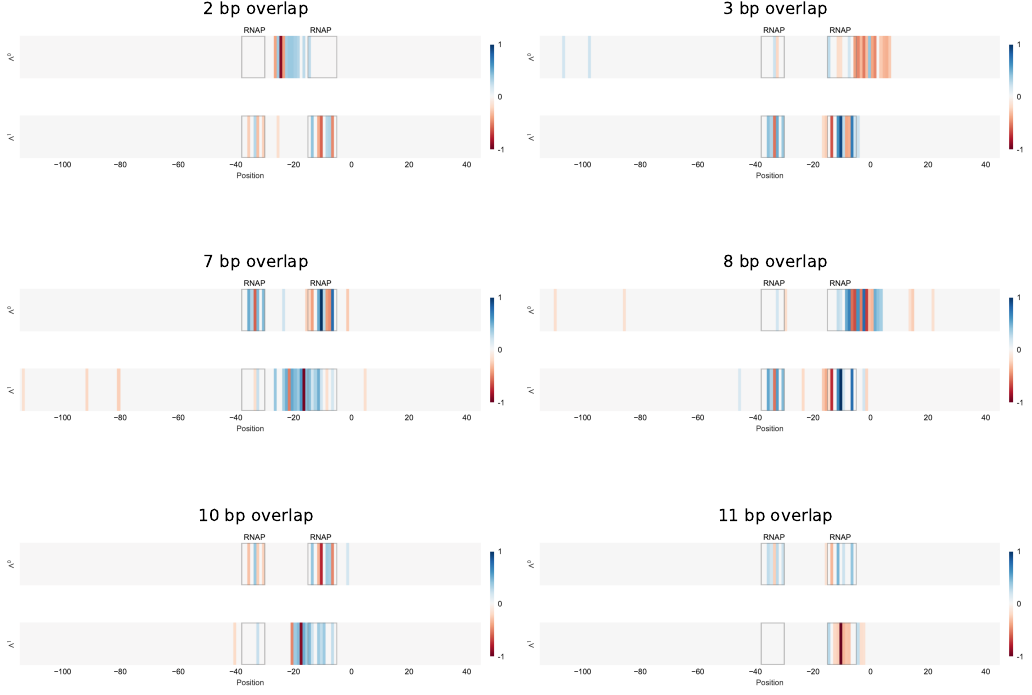}
    \caption{\textbf{Systematic survey of overlapping RNAP and repressor binding sites using synthetic data.}
        Six promoter architectures with progressively increasing overlap (2--11\,bp) between the RNAP and repressor binding sites. In each panel, the two-component blueprint ($n=2$) resolves the overlapping regulatory elements into distinct filters, even when the sites share a substantial fraction of their positions.}
    \label{fig:overlap_grid}
\end{figure}

Figure~\ref{fig:overlap_grid} presents six additional examples where the repressor binding site overlaps with the RNAP binding site to varying widths (from 2 to 11\,bp), demonstrating that the blueprint manages to resolve the binding sites across a range of overlap geometries.

\subsection{Double repression, logic gates, and DNA looping} 
We now consider how the possible interactions between multiple TFs, such as cooperative binding, logic-gate architecture, and DNA looping, can be captured by the optimal compression scheme.

\begin{figure}[ht!]
    \centering
    \includegraphics[width=0.9\columnwidth]{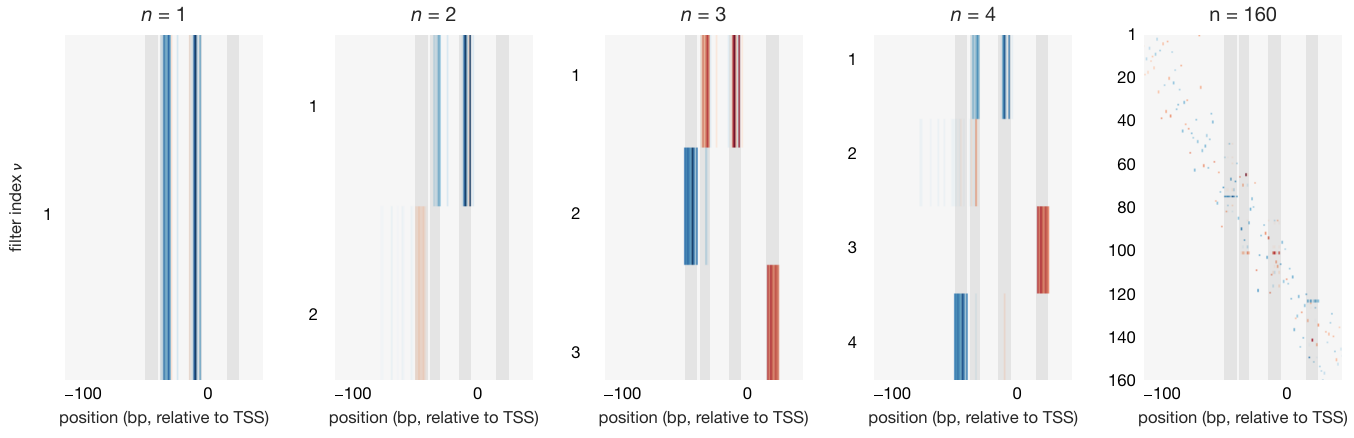}
    \caption{\textbf{Informational blueprint filters for the double repression architecture for various number of hyperletters.} 
    The three binding sites for RNAP and the two repressors are captured in separate filters when the number of hyperletters is at least three.
    When the number of hyperletters is increased way above the number of binding sites, here shown for $n=160$, the filters start to resolve the individual base pairs which can help capturing microscopic correlations with the further decimal points of expression.
    }
    \label{fig:double_repression}
\end{figure}

\subsubsection{Logic gates determining cooperation between transcription factors}
Consider a double repression architecture, where two repressors R$_1$ and R$_2$ regulate the same promoter. The regulatory logic---i.e.\ whether both repressors must bind to silence the gene (\texttt{AND} logic) or either one suffices (\texttt{OR} logic)---leaves a distinct structural signature in the blueprint filters.
Figure~\ref{fig:double_repression} shows the \texttt{OR}-logic case, where the two repressors act independently. When $n=1$, both repressor sites are conflated into a single filter; stepping to $n=2$ cleanly separates them into distinct filters, each localising on its respective binding site.

\begin{figure}[hbt!]
    \centering
    \includegraphics[width=0.4\columnwidth]{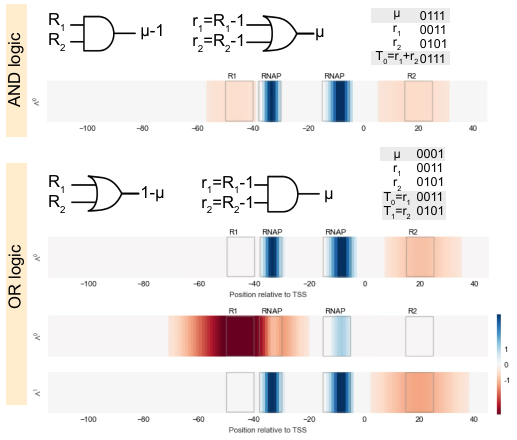}
    \caption{\textbf{Reading off regulatory logic from filter structure in synthetic data.}
        Two repressors R$_1$ and R$_2$ can combine via different logic gates. Because filters couple to mutations, we define $r_i$: the variable $r_i = 1$ when repressor $i$ \emph{cannot} bind. Top (\texttt{AND} logic): Repression requires both repressors; expression is high if \emph{either} site is disrupted ($r_1$ \texttt{OR} $r_2$). A single filter computes $T_0 = r_1 + r_2$, which distinguishes all functionally distinct states (truth table, right). The filter couples to both repressor sites with the same sign. Bottom (\texttt{OR} logic): Either repressor alone suffices; expression is high only if \emph{both} sites are disrupted ($r_1$ \texttt{AND} $r_2$). Since a single linear filter cannot compute \texttt{AND}, we need two filters, each coupling to a distinct repressor site. The two-bit representation separates all functionally distinct configurations.
    }
    \label{fig:logic}
\end{figure}

This is a direct consequence of the cooperative logic. Under \texttt{AND} logic, repression requires both repressors to bind; disrupting \emph{either} site alone is sufficient to relieve repression. A single linear filter $T_0 = r_1 + r_2$ (where $r_i = 1$ when site $i$ is disrupted) therefore distinguishes all functionally distinct expression states. In contrast, under \texttt{OR} logic, either repressor alone suffices for repression, and expression is high only when \emph{both} sites are disrupted---a condition that a single linear filter cannot express. The blueprint responds by assigning each repressor its own filter, yielding a two-bit representation. Figure~\ref{fig:logic} contrasts the two cases, including the truth tables that justify the minimal number of filters required.

An alternative way to distinguish cooperative from independent regulation is to look at the information footprint signal at each site as a function of TF copy number~\cite{Pan2024}.
Under \texttt{AND} logic, the signal at both sites vanishes simultaneously when one repressor is absent, whereas under \texttt{OR} logic each site retains its signal independently. The blueprint filters capture the distinction between cooperative and independent regulation structurally, without requiring titration experiments at multiple copy numbers.

\subsubsection{DNA looping}\label{app:dna_looping}
The double repression architecture also provides the setting for DNA looping, a fundamentally non-local regulatory mechanism. In the \emph{lac} operon of \emph{E.~coli}, a LacI tetramer simultaneously occupies two spatially separated operator sites (O$_1$ and O$_3$, or O$_1$ and O$_2$), forcing the intervening DNA to bend into a loop. This looped configuration sterically occludes RNAP, providing an additional repressive state beyond simple single-operator binding.

In the thermodynamic model, the looped state contributes a Boltzmann weight that depends on (i) the binding energies at both operator sites and (ii) the free energy cost of bending the DNA into a loop, $\Delta F_{\rm loop}$. Crucially, the two operators do not function independently: disrupting \emph{either} operator site abolishes looping entirely, so a mutation at one site affects the regulatory contribution of the distant site. This strong mechanical cooperativity means that, from the standpoint of mutual information, both operators behave as a single non-local regulatory unit.

Figure~\ref{fig:dna_looping} shows the three-component compression map for a synthetic looping architecture. Despite using $n=3$ filters, the optimisation does \emph{not} separate the two operator sites into distinct components. Instead, a single filter ($\Lambda_2$) couples to both distant operators simultaneously, while $\Lambda_1$ captures RNAP. The third filter remains essentially empty, confirming that only two functional regulatory elements are present---RNAP and the cooperative repressor dimer---even though the repressor occupies two spatially separated stretches of DNA. This contrasts sharply with the independent double repression case (Fig.~\ref{fig:double_repression}), where each repressor site is assigned its own filter. The number of non-trivial filters thus distinguishes cooperative from independent multi-site regulation.

\begin{figure}[ht!]
    \centering
    \includegraphics[width=0.9\columnwidth]{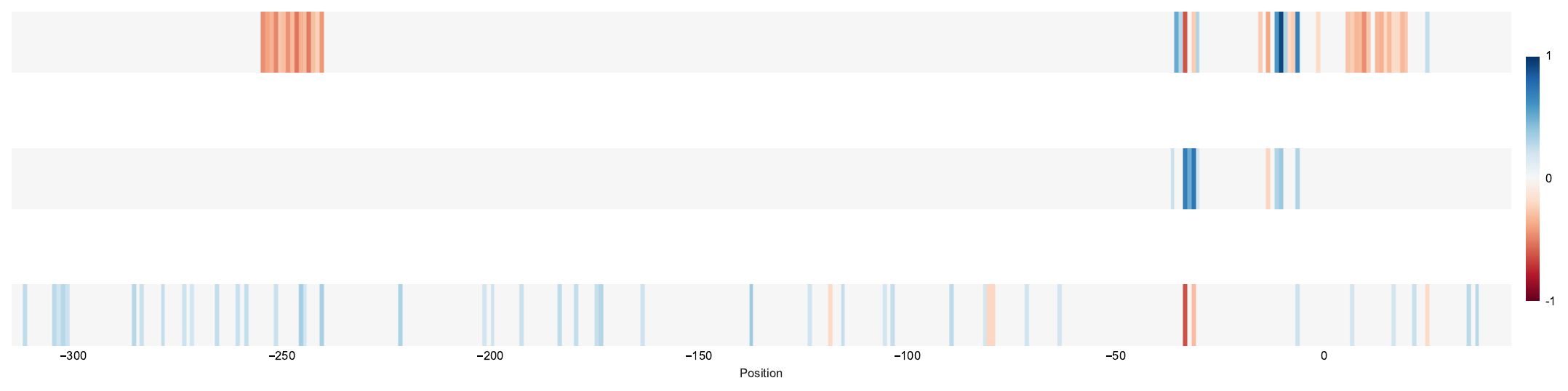}
    \caption{\textbf{Filter structure for DNA looping.} The three-component optimal compression map for a regulatory architecture where a LacI tetramer binds two distant operator sites, forcing the DNA to form a loop. A single filter ($\Lambda_2$) couples to both operator sites simultaneously, reflecting their cooperative function as one non-local regulatory unit. The third filter ($\Lambda_3$) is essentially empty, confirming that the system contains only two functional regulatory elements despite two spatially separated repressor binding sites.}
    \label{fig:dna_looping}
\end{figure}

The filter merging observed in Figure~\ref{fig:dna_looping} can be understood quantitatively by applying the decomposition of Section~\ref{sec:covariance} to the looping partition function. With RNAP and a LacI tetramer that can occupy two operator sites O$_1$ and O$_2$, the relevant promoter states are: empty, RNAP bound, repressor at O$_1$ alone, repressor at O$_2$ alone, and the looped state where the repressor simultaneously occupies both operators. The partition function reads
\begin{equation}\label{eq:Z_loop}
  Z = 1 + \frac{p}{N_{\mathrm{NS}}}\,e^{-\beta E_P} + \frac{r}{N_{\mathrm{NS}}}\,e^{-\beta E_{R_1}} + \frac{r}{N_{\mathrm{NS}}}\,e^{-\beta E_{R_2}} + \frac{r^2}{N_{\mathrm{NS}}^2}\,e^{-\beta(E_{R_1}+E_{R_2}+\Delta F_{\mathrm{loop}})},
\end{equation}
where $E_{R_1}$ and $E_{R_2}$ are the binding energies at the two operator sites and $\Delta F_{\mathrm{loop}}<0$ is the free energy gain from loop formation. Expression is proportional to the RNAP occupancy, $\mu\propto p_{\mathrm{on}} =(p/N_{\mathrm{NS}})\,e^{-\beta E_P}/Z$.
 
Since each binding energy $E_a$ ($a\in\{P,R_1,R_2\}$) depends only on mutations within its own site, Eq.~\eqref{eq:cov_decomp} gives
\begin{equation}\label{eq:cov_loop}
  C_i \;=\; \chi_P\;\delta\bar{E}_i^{(P)}\;\mathbf{1}_{i\in P} \;+\; \chi_{R_1}\;\delta\bar{E}_i^{(R_1)}\;\mathbf{1}_{i\in\mathrm{O}_1} \;+\; \chi_{R_2}\;\delta\bar{E}_i^{(R_2)}\;\mathbf{1}_{i\in\mathrm{O}_2},
\end{equation}
where $\chi_a=\langle\partial\mu/\partial E_a\rangle_{\mathrm{eff}}$ are thermodynamic susceptibilities and $\delta\bar{E}_i^{(a)}$ is the mutation-averaged energy shift (Eq.~\ref{eq:deltaE}). The structure of these susceptibilities depends critically on $\Delta F_{\mathrm{loop}}$.
 
When the looping energy is negligible ($\Delta F_{\mathrm{loop}}\to 0$), the looped state is suppressed and the two operators act as independent repressors. In this limit $\chi_{R_1}$ and $\chi_{R_2}$ are determined by distinct terms in the partition function and are only weakly correlated across the mutant ensemble. The covariance vector $C_i$ then has two independent components supported on O$_1$ and O$_2$, and the IB optimisation assigns each operator its own filter---precisely the independent double-repression scenario of Figure~\ref{fig:double_repression}.
 
When looping is strongly favourable ($|\Delta F_{\mathrm{loop}}|\gg k_BT$), the looped Boltzmann weight in Eq.~\eqref{eq:Z_loop} dominates the repressive states. Since this weight depends on $E_{R_1}+E_{R_2}$, the susceptibilities become locked:
\begin{equation}\label{eq:chi_lock}
  \chi_{R_1} \;\approx\; \chi_{R_2} \;\approx\; -\Bigl\langle\frac{\partial p_{\mathrm{on}}}{\partial(E_{R_1}+E_{R_2})}\Bigr\rangle_{\!\mathrm{eff}}.
\end{equation}
Equation~\eqref{eq:cov_loop} then collapses to two independent directions: one along the RNAP site and one along $\delta\bar{E}_i^{(R_1)}\,\mathbf{1}_{i\in\mathrm{O}_1} +\delta\bar{E}_i^{(R_2)}\,\mathbf{1}_{i\in\mathrm{O}_2}$, spanning both operators simultaneously. Since $\Lambda_i\propto C_i$ (Section~\ref{sec:covariance}), the optimal filter aligns with this merged direction, explaining why $\Lambda_2$ in Figure~\ref{fig:dna_looping} couples to both distant operator sites as a single non-local regulatory unit.
 
The degree of locking between $\chi_{R_1}$ and $\chi_{R_2}$ therefore provides a continuous interpolation between independent multi-site repression (separate filters) and cooperative looping (merged filter), controlled by a single physical parameter $\Delta F_{\mathrm{loop}}$.

\section{Proof of the variational representation of mutual information}
Let us prove the Donsker-Varadhan formula in Eq.~\ref{eq:DV}. We will first prove that the expression in the supremum is strictly a lower bound of mutual information, and then show that the supremum is attained by a suitable choice of $f(a,b)$ and the value of the supremum is equal to the mutual information.

\begin{lemma}[Lower bound of mutual information]
    For any function $f(a,b) \in \mathbb{R}$, the following inequality holds
    \begin{equation}\label{eq:DVlower_bound}
        I(A:B) \geq \mathbb{E}_{P(a,b)}\left[f(a,b)\right] - \log \mathbb{E}_{P(a)P(b)}\left[e^{f(a,b)}\right].
    \end{equation}
\end{lemma}

\begin{proof}
    We need to show that
    \begin{align}\label{eq:DVlower_bound_proof1}
        0 \leq& I(A:B) -  \mathbb{E}_{P(a,b)}\left[f(a,b)\right] + \log \mathbb{E}_{P(a)P(b)}\left[e^{f(a,b)}\right]\nonumber\\
        &=\mathbb{E}_{P(a,b)}\left[\log \frac{P(a,b)}{P(a)P(b)} - f(a,b)\right] + \log \mathbb{E}_{P(a)P(b)}\left[e^{f(a,b)}\right]\nonumber\\
        &=\mathbb{E}_{P(a,b)}\left[\log \frac{P(a,b)}{P(a)P(b) e^{f(a,b)}} \right] + \log \mathbb{E}_{P(a)P(b)}\left[e^{f(a,b)}\right]\nonumber\\
        &=\mathbb{E}_{P(a,b)}\left[\log \frac{P(a,b)}{P(a)P(b) e^{f(a,b)}}  + \log \mathbb{E}_{P(a)P(b)}\left[e^{f(a,b)}\right]\right]\nonumber\\
        &=\mathbb{E}_{P(a,b)}\left[\log \frac{P(a,b)\mathbb{E}_{P(a)P(b)}\left[e^{f(a,b)}\right]}{P(a)P(b) e^{f(a,b)}} \right].
    \end{align}
    Let us now define a new probability distribution:
    \begin{equation}
        Q(a,b) := \frac{P(a)P(b) e^{f(a,b)}}{\mathbb{E}_{P(a)P(b)}\left[e^{f(a,b)}\right]}.
    \end{equation}
    After substituting this expression in \ref{eq:DVlower_bound_proof1}, what is left is to prove that
    \begin{equation}
        0 \leq \mathbb{E}_{P(a,b)}\left[\log \frac{P(a,b)}{Q(a,b)} \right].
    \end{equation}

    We can show this using Jensen's inequality:
    \begin{equation}\label{eq:jensen}
        \varphi \left(\mathbb{E}_{P(x)}[f(x)]\right) \geq \mathbb{E}_{P(x)}\left[\varphi(f(x)),\right]
    \end{equation}
    which holds for any $f(x)$ and any concave function $\varphi$. Because $\log$ is a strictly concave function, we have
    \begin{align*}
        \mathbb{E}_{P(a,b)}\left[\log \frac{P(a,b)}{Q(a,b)} \right] =& - \mathbb{E}_{P(a,b)}\left[\log \frac{Q(a,b)}{P(a,b)} \right]\\
        \geq& -\log \mathbb{E}_{P(a,b)}\left[\frac{Q(a,b)}{P(a,b)} \right]\\
        &= \log 1 = 0.
    \end{align*}
\end{proof}

\begin{lemma}[Existence of supremum in DV representation]
    There exists a function $f^*(a,b) \in \mathbb{R}$, such that 
    \begin{equation}
        I(A:B) = \mathbb{E}_{P(a,b)}\left[f^*(a,b)\right] - \log \mathbb{E}_{P(a)P(b)}\left[e^{f^*(a,b)}\right].
    \end{equation}
\end{lemma}

\begin{proof}
    Let us choose $f^*(a,b) = \log \frac{P(a,b)}{P(a)P(b)}$. 
    Then we have
    \begin{align*}
        \mathbb{E}_{P(a,b)}\left[f^*(a,b)\right] - \log \mathbb{E}_{P(a)P(b)}\left[e^{f^*(a,b)}\right] &= \mathbb{E}_{P(a,b)}\left[\log \frac{P(a,b)}{P(a)P(b)}\right] - \log \mathbb{E}_{P(a)P(b)}\left[e^{\log \frac{P(a,b)}{P(a)P(b)}}\right]\\
        &= I(A:B) - \log\mathbb{E}_{P(a)P(b)}\left[\frac{P(a,b)}{P(a)P(b)}\right]\\
        &= I(A:B) - \sum_{a\in \mathcal{A}}\sum_{ b \in \mathcal{B}} P(a)P(b) \frac{P(a,b)}{P(a)P(b)}\\
        &= I(A:B) - \log 1 = I(A:B),
    \end{align*}
    where the last line follows from the normalisation of the joint distribution $P(a,b)$.
\end{proof}

\end{document}